\documentclass[a4paper,UKenglish,numberwithinsect]{lipics-v2019}

\usepackage{microtype}
\usepackage{xspace}
\usepackage{proof}
\usepackage[only,llbracket,rrbracket]{stmaryrd}
\SetSymbolFont{stmry}{bold}{U}{stmry}{m}{n} 
\usepackage{bbm}
\usepackage{cite}
\usepackage{tikz}

\newcommand\xrecap[4]{{\noindent\textcolor{darkgray}{$\blacktriangleright$}\nobreakspace\sffamily\bfseries}{\bf #1 \ref{#3}}{ (#2)}{\bf.} \emph{#4}}
\newcommand\recap[3]{{\noindent\textcolor{darkgray}{$\blacktriangleright$}\nobreakspace\sffamily\bfseries}{\bf #1 \ref{#2}.} \emph{#3}}
\newcommand\T{\ensuremath{{\cal T}}}
\newcommand\hole[1]{\ensuremath{[]^{#1}}}
\newcommand\s[1]{\ensuremath{s({#1})}}
\newcommand\conj[1]{\ensuremath{\mathsf{conj}(#1)}}
\newcommand\PF[1]{\ensuremath{\mathsf{PF}(#1)}}
\newcommand\OC{System I\xspace}

\newcommand\interp[1]{\ensuremath{\left\llbracket{#1}\right\rrbracket}}

\newcommand\SN{\ensuremath{\mathsf{SN}}}

\newcommand\eq{\ensuremath{\rightleftarrows}}
\newcommand\re{\ensuremath{\hookrightarrow}}
\newcommand\toreq{\ensuremath{\rightsquigarrow}}
\newcommand\tf{\mbox{\bf T\!F}}
\newcommand\true{\mbox{\bf T}}
\newcommand\false{\mbox{\bf F}}
\newcommand\trtr{\mbox{\bf t\!t}}
\newcommand\ff{\mbox{\bf f\!f}}
\newcommand\ve[1]{#1}
\newcommand\cond[1]{\ensuremath{\scriptstyle[#1]\,}}
\newcommand\condi[2]{\ensuremath{{#1}{\vcenter{#2}}}}
\newcommand\rulelabel[1]{\mbox{\scriptsize\sc{#1}}}

\newcommand\numbertype[1]{\ensuremath{\mathbbm{#1}}}

\newcommand\nd{\ensuremath{\oplus}}

\newcommand\V{\mathcal V}

\title{Proof Normalisation in a Logic Identifying Isomorphic Propositions}

\author{Alejandro D\'iaz-Caro}{Instituto
  de Ciencias de la Computaci\'on (CONICET-Universidad de Buenos Aires) \and Universidad Nacional de Quilmes,
Argentina}{adiazcaro@icc.fcen.uba.ar}{https://orcid.org/0000-0002-5175-6882}{}
\author{Gilles Dowek}{Inria, LSV, ENS Paris-Saclay,
France}{gilles.dowek@ens-paris-saclay.fr}{https://orcid.org/0000-0001-6253-935X}{}

\authorrunning{A. D\'iaz-Caro and G. Dowek}

\Copyright{Alejandro D\'iaz-Caro and Gilles Dowek}



\category{}

\relatedversion{A variant of the system presented in this paper, has
  been published in
  \url{https://arxiv.org/abs/1303.7334v1} (LSFA 2012 \cite{DiazcaroDowekLSFA12}), with a sketch of
  a subject reduction proof but no normalisation or consistency proofs.}

\supplement{}

\funding{Partially supported by ECOS-Sud A17C03, PICT 2015-1208, and the French-Argentinian laboratory SINFIN.}

\EventEditors{Herman Geuvers}
\EventNoEds{1}
\EventLongTitle{4th International Conference on Formal Structures for Computation and Deduction (FSCD 2019)}
\EventShortTitle{FSCD 2019}
\EventAcronym{FSCD}
\EventYear{2019}
\EventDate{June 24--30, 2019}
\EventLocation{Dortmund, Germany}
\EventLogo{}
\SeriesVolume{131}
\ArticleNo{14}
\nolinenumbers 
\hideLIPIcs  

\begin{document}

\maketitle

\begin{abstract}
  We define a fragment of propositional logic where isomorphic propositions,
  such as $A\wedge B$ and $B\wedge A$, or $A\Rightarrow(B\wedge
  C)$ and $(A\Rightarrow B)\wedge(A\Rightarrow C)$ are identified.
  We define \OC, a proof language for this logic, and prove its normalisation and
  consistency.
 \end{abstract}

\section{Introduction}\label{sec:intro}

\subsection{Identifying isomorphic propositions}

In mathematics, addition is associative and commutative, multiplication
distributes over addition, etc. In contrast, in logic conjunction is neither associative
nor commutative, implication does not distribute over conjunction, etc. For
instance, the propositions $A \wedge B$ and $B \wedge A$ are different: if $A
\wedge B$ has a proof, then so does $B \wedge A$, but if $\ve r$ is a proof of
$A \wedge B$, then it is not a proof of $B \wedge A$.

A first step towards considering $A \wedge B$ and $B \wedge A$ as the same
proposition has been made in
\cite{RittriCADE90,BruceDiCosmoLongoMSCS92,DiCosmo95,DiCosmoMSCS05}, where a
notion of isomorphic propositions has been defined: two propositions $A$ and $B$
are isomorphic if there exist two proofs of $A \Rightarrow B$ and $B
\Rightarrow A$ whose composition, in both ways, is the identity.

For the fragment of propositional logic restricted to the operations
$\Rightarrow$ and $\wedge$, all the isomorphisms are consequences of the
following four:
\begin{align}
  A\wedge B		  &\equiv B\wedge A \label{iso:comm}\\
  A\wedge (B\wedge C)	  &\equiv (A\wedge B)\wedge C \label{iso:asso}\\
  A\Rightarrow(B\wedge C)  &\equiv (A\Rightarrow B)\wedge(A\Rightarrow C) \label{iso:distrib}\\
  (A\wedge B)\Rightarrow C &\equiv A\Rightarrow B\Rightarrow C \label{iso:curry}
\end{align}
For example, $(A\Rightarrow B\Rightarrow C)\equiv (B\Rightarrow A\Rightarrow C)$
is a consequence of \eqref{iso:curry}  and \eqref{iso:comm} \cite{BruceDiCosmoLongoMSCS92}.

In this paper, we go one step further and define a proof language, \OC, for the
fragment $\Rightarrow$, $\wedge$, such that when $A \equiv B$, then any proof of
$A$ is also a proof of $B$, so the propositions $A \wedge B$ and $B \wedge A$,
for instance, are really identical, as they have the same proofs.

The idea of identifying some propositions has already been investigated, for
example, in Martin-L\"of's type theory \cite{MartinLof84}, in the Calculus of
Constructions \cite{CoquandHuetIC88}, and in Deduction modulo theory
\cite{DowekHardinKirchnerJAR03,DowekWernerJSL98}, where definitionally
equivalent propositions, for instance $A \subseteq B$, $A \in\mathcal P(B)$, and
$\forall x~(x \in A \Rightarrow x \in B)$ can be identified. But definitional
equality does not handle isomorphisms. For example, $A \wedge B$ and
$B \wedge A$ are not identified in these logics.
Beside definitional equality, identifying isomorphic types in type theory, is
also a goal of the univalence axiom \cite{HoTT}.

Isomorphisms make proofs more natural. For instance, to prove
$(A\wedge(A\Rightarrow B))\Rightarrow B$ in natural deduction we need to
introduce conjunctive hypothesis $A\wedge (A\Rightarrow B)$ which has to be
decomposed into $A$ and $A\Rightarrow B$, while using the
isomorphism~\eqref{iso:curry} allows to transform the goal to $A\Rightarrow
(A\Rightarrow B)\Rightarrow B$ and introduce directly the hypotheses $A$ and
$A\Rightarrow B$, eliminating completely the need for conjunctive hypotheses.

\subsection{Lambda-calculus}

The proof-language of the fragment of propositional logic restricted to the
operations $\Rightarrow$ and $\wedge$ is simply typed lambda-calculus extended
with Cartesian product. So, \OC is an extension of this calculus
where, for example, a pair of functions $\langle\ve r,\ve s \rangle$ of type
$(A\Rightarrow B)\wedge(A\Rightarrow C)\equiv A\Rightarrow(B\wedge C)$ can be
applied to an argument $\ve t$ of type $A$, yielding a term $\langle\ve r,\ve
s\rangle\ve t$ of type $B\wedge C$. For example, the term $\langle \lambda x^\tau.x,\lambda x^\tau.x
\rangle y$ has type $\tau\wedge\tau$. With the usual reduction rules of lambda
calculus with pairs, such a term would be normal, but we can also extend the
reduction relation, with an equation $\langle r,s \rangle t\eq \langle rt,st
\rangle$, such that this term is equivalent to $\langle (\lambda
x^\tau.x)y,(\lambda x^\tau.x)y \rangle$ and thus reduces to $\langle
y,y\rangle$. Taking too many of such equations may lead to non termination
(Section~\ref{sec:deltadelta}), and taking too few multiplies undesired normal
forms. The choice of the rules in this paper is motivated by the goal to have
both termination of reduction (Section~\ref{sec:SN}) and consistency
(Section~\ref{sec:cons}), that is, no normal closed term of atomic types.

To stress the associativity and commutativity of the
notion of pair, we write $\ve r\times\ve s$ instead of $\langle\ve r,\ve s\rangle$ and thus write this equivalence as
\[
  (\ve r\times\ve s)\ve t\eq\ve r\ve t\times\ve s\ve t
\]
Several similar equivalence rules on terms are introduced: one related to the
isomorphism \eqref{iso:comm}, the commutativity of the conjunction, $\ve r\times
\ve s\eq \ve s\times \ve r$; one related to the isomorphism \eqref{iso:asso},
the associativity of the conjunction, $ (\ve r\times \ve s)\times \ve t\eq \ve
r\times (\ve s\times \ve t) $; two to the isomorphism \eqref{iso:distrib}, the
distributivity of implication with respect to conjunction,
$\lambda x.(r\times s)\eq\lambda x.r\times\lambda x.s$
and
$ (\ve r\times \ve s)\ve t\eq \ve r\ve t\times \ve s\ve t $;
and one related to the isomorphism \eqref{iso:curry}, the currification, $\ve r\ve s\ve t\eq\ve r(\ve
s\times \ve t)$.

One of the difficulties in the design of \OC is the design of the
elimination rule for the conjunction. A rule like ``if $\ve r:A\wedge B$ then $\pi_1(\ve r):A$'', would not be consistent.
Indeed, if $A$ and $B$ are two arbitrary types, $\ve s$ a term of type $A$ and
$\ve t$ a term of type $B$, then $\ve s\times \ve t$ has both type $A\wedge B$
and type $B \wedge A$, thus $\pi_1(\ve s \times \ve t)$ would have both type $A$ and
type $B$. A solution is to consider explicitly typed (Church
style) terms, and parametrise the projection by the type: if $\ve
r:A\wedge B$ then $\pi_A(\ve r):A$ and the reduction rule is then
that $\pi_A(\ve s \times \ve t)$ reduces to $\ve s$ if $\ve s$ has type $A$.

This rule makes reduction non-deterministic. Indeed, in the particular case
where $A$ happens to be equal to $B$, then both $\ve s$ and $\ve t$ have type
$A$ and $\pi_A(\ve s \times \ve t)$ reduces both to $\ve s$ and to $\ve t$.
Notice that, although this reduction rule is non-deterministic, it preserves
typing, like the calculus developed in \cite{DowekJiangIC11}, where the
reduction is non-deterministic, but verifies subject reduction.

\subsection{Non-determinism}

Therefore, \OC is one of the many non-deterministic calculi in the sense, for
instance,
of~\cite{BoudolIC94,BucciarelliEhrhardManzonettoAPAL12,deLiguoroPipernoIC95,DezaniciancagliniDeliguoroPipernoSIAM98,PaganiRonchidellaroccaFI10}
and our pair-construction operator $\times$ is also the parallel composition
operator of a non-deterministic calculus.

In non-deterministic calculi, the non-deterministic choice is such that if $\ve
r$ and $\ve s$ are two $\lambda$-terms, the term $\ve r\nd \ve s$ represents the
computation that runs either $\ve r$ or $\ve s$ non-deterministically, that is
such that $(\ve r \nd \ve s)\ve t$ reduces either to $\ve r\ve t$ or $\ve s\ve
t$. On the other hand, the parallel composition operator $|$ is such that the
term $(\ve r~|~\ve s)\ve t$ reduces to $\ve r\ve t~|~\ve s\ve t$ and continue
running both $\ve r\ve t$ and $\ve s\ve t$ in parallel. In our case, given $\ve
r$ and $\ve s$ of type $A\Rightarrow B$ and $\ve t$ of type $A$, the term
$\pi_B((\ve r\times\ve s)\ve t)$ is equivalent to $\pi_B(\ve r\ve t\times\ve
s\ve t)$, which reduces to $\ve r\ve t$ or $\ve s\ve t$, while the term $\ve
r\ve t\times\ve s\ve t$ itself would run both computations in parallel.
Hence, our $\times$ is equivalent to the parallel composition while the
non-deterministic choice $\oplus$ is decomposed into $\times$ followed by $\pi$.

In \OC, the non-determinism comes from the interaction of two operators,
$\times$ and $\pi$. This is similar to quantum computing where the
non-determinism comes from the interaction of two operators, the fist allowing
to build a superposition, that is a linear combination, of two terms $\alpha.\ve
r+\beta.\ve t$, and the measurement operator $\pi$. In addition, in such
calculi, the distributivity rule $(\ve r+\ve s)\ve t\eq\ve r\ve t+\ve s\ve t$ is
seen as the point-wise definition of the sum of two functions.

More generally, the calculus developed in this paper is also related to the
algebraic
calculi~\cite{ArrighiDiazcaroLMCS12,ArrighiDiazcaroValironIC17,ArrighiDowekRTA08,
ArrighiDowekLMCS17,VauxMSCS09,DiazcaroPetitWoLLIC12,DiazcaroDowekTPNC17}, some
of which have been designed to express quantum algorithms. There is a clear link
between the pair constructor $\times$ and the projection $\pi$, with the
superposition constructor $+$ and the measurement $\pi$ on these calculi. In
these cases, the pair $\ve s + \ve t$ is not interpreted as a non-deterministic
choice, but as a superposition of two processes running $\ve s$ and $\ve t$, and
the operator $\pi$ is the projection related to the measurement, which is the
only non-deterministic operator.

\subsection*{Outline}
In Section~\ref{sec:eqTypes}, we define the notion of type isomorphism and prove
elementary properties of this relation.
In Section~\ref{sec:calculus}, we introduce \OC.
In Section~\ref{sec:SR}, we prove its subject reduction.
In Section~\ref{sec:SN}, we prove its strong normalisation.
In Section~\ref{sec:cons}, we prove its consistency.
Finally, in Section~\ref{sec:computing}, we discuss how \OC could be used as a programming language.

\section{Type isomorphisms}\label{sec:eqTypes}
\subsection{Types and isomorphisms}
Types are defined by the following grammar
\[
  A,B,C,\dots\ ::=\quad \tau~|~A\Rightarrow B~|~A\wedge B
\]
where $\tau$ is the only atomic type.
\begin{definition}[Size of a type]
  The size of a type is defined as usual by
  \begin{align*}
    \s{\tau} &= 1\\ 
    \s{A\Rightarrow B} &=\s{A}+\s{B}+1\\
    \s{A\wedge B} &=\s{A}+\s{B}+1
  \end{align*}
\end{definition}

\begin{definition}
  [Congruence]
The isomorphisms \eqref{iso:comm}, \eqref{iso:asso}, \eqref{iso:distrib}, and
\eqref{iso:curry} define a congruence on types.
\begin{align*}
  A\wedge B		  &\equiv B\wedge A \tag{1}\\
  A\wedge (B\wedge C)	  &\equiv (A\wedge B)\wedge C \tag{2}\\
  A\Rightarrow(B\wedge C)  &\equiv (A\Rightarrow B)\wedge(A\Rightarrow C) \tag{3}\\
  (A\wedge B)\Rightarrow C &\equiv A\Rightarrow B\Rightarrow C \tag{4}
\end{align*}
\end{definition}

\subsection{Prime factors}\label{sec:PF}
\begin{definition}[Prime types]
  A prime type is a type of the form $C_1\Rightarrow\dots\Rightarrow C_n\Rightarrow\tau$, with $n\geq 0$.
\end{definition}

A prime type is equivalent to $(\bigwedge_{i=1}^nC_i)\Rightarrow\tau$, which is
either equivalent to $\tau$ or to $C\Rightarrow\tau$, for some $C$.
For uniformity, we may write
$\emptyset\Rightarrow\tau$ for $\tau$.

We now show that each type can be decomposed into a conjunction of prime types.
We use the notation $[A_i]_{i=1}^n$ for the multiset whose elements are $A_1,\dots,A_n$. We may write $[A_i]_i$ when the number of elements is not important.
If $R=[A_i]_i$ is a multiset of types, then $\conj R=\bigwedge_i A_i$.
\begin{definition}
  We write $[A_1,\dots,A_n]\sim[B_1,\dots,B_m]$ if $n=m$ and $B_i\equiv A_i$.
\end{definition}
\begin{definition}[Prime factors]
  The multiset of prime factors of a type $A$ is inductively defined as follows
  \begin{align*}
    \PF\tau &=[\tau]\\
    \PF{A\Rightarrow B} &= [(A\wedge C_i)\Rightarrow\tau]_{i=1}^n
                          \quad\textrm{ where }[C_i\Rightarrow\tau]_{i=1}^n=\PF B\\
    \PF{A\wedge B} &=\PF A\uplus\PF B
  \end{align*}
  with the convention that $A\wedge\emptyset=A$.
\end{definition}
Note that if $B\Rightarrow\tau\in\PF A$, then $s(B)<s(A)$.

\begin{lemma}\label{lem:eqConjPF}
  For all $A$, $A\equiv\conj{\PF A}$. 
\end{lemma}
\begin{proof}
  By induction on \s A.
  \begin{itemize}
  \item If $A=\tau$, then $\PF\tau=[\tau]$, and so $\conj{\PF\tau}=\tau$.
  \item If $A=B\Rightarrow C$, then $\PF A = [(B\wedge C_i)\Rightarrow\tau]_i$,
    where $[C_i\Rightarrow\tau]_i=\PF C$. By the induction hypothesis, $C\equiv\bigwedge_i(C_i\Rightarrow\tau)$,
    hence, $A = B\Rightarrow C\equiv B\Rightarrow\bigwedge_i (C_i\Rightarrow\tau)\equiv
    \bigwedge_i(B\Rightarrow C_i\Rightarrow\tau)\equiv\bigwedge_i((B\wedge C_i)\Rightarrow\tau)$.
  \item If $A=B\wedge C$, then $\PF{A}=\PF{B}\uplus\PF{C}$. By the induction
    hypothesis, $B\equiv\conj{\PF B}$,
    and $C\equiv\conj{\PF C}$.
    Therefore,
    $A
    =B\wedge C
    \equiv \conj{\PF B}\wedge\conj{\PF C}\equiv\conj{\PF{B\wedge C}}
    \equiv\conj{\PF B\uplus\PF C}
    =\conj{\PF A}$.
    \qedhere
  \end{itemize}
\end{proof}

\begin{lemma}\label{lem:eqPF}
  If $A\equiv B$, then $\PF A\sim\PF B$.
\end{lemma}
\begin{proof}
  First we check that $\PF{A\wedge B}\sim\PF{B\wedge A}$ and similar for the other three isomorphisms.
  Then we prove by structural induction that if $A$ and $B$ are equivalent in
  one step, then $\PF A\sim\PF B$.
  We conclude by an induction on the length of the derivation of the
  equivalence $A\equiv B$.
\end{proof}

\subsection{Measure of types}
The size of a type is not preserved by equivalence.
For instance, $
\tau\Rightarrow(\tau\wedge\tau)
\equiv
(\tau\Rightarrow\tau)\wedge(\tau\Rightarrow\tau)
$, but
$\s{\tau\Rightarrow(\tau\wedge\tau)}=5$ and $
\s{(\tau\Rightarrow\tau)\wedge(\tau\Rightarrow\tau)}=7$. Thus, we define
another notion of measure of a type.

\begin{definition}[Measure of a type]
The measure of a type is defined as follows
\[
  m(A) = \sum_i(m(C_i)+1)\quad\textrm{ where }[C_i\Rightarrow\tau]_i=\PF A
\]
with the convention that $m(\emptyset)=0$.
\end{definition}

\begin{lemma}
  If $A\equiv B$, then $m(A)=m(B)$.
\end{lemma}
\begin{proof}
  By induction on $s(A)$. Let $\PF A = [C_i\Rightarrow\tau]_i$ and $\PF B =
[D_j\Rightarrow\tau]_j$. By Lemma~\ref{lem:eqPF},
$[C_i\Rightarrow\tau]_i\sim[D_i\Rightarrow\tau]_i$. Without lost of generality,
take $C_i\equiv D_i$. By the induction hypothesis, $m(C_i)=m(D_i)$. Then,
$m(A)=\sum_i(m(C_i)+1)=\sum_i(m(D_i)+1)=m(B)$.
\end{proof}

The following lemma shows that the measure $m(A)$ verifies the usual
properties.

\begin{lemma}
  \label{lem:comp}~
  \begin{enumerate}
  \item $m(A\wedge B)>m(A)$
  \item $m(A\Rightarrow B)>m(A)$
  \item $m(A\Rightarrow B)>m(B)$
  \end{enumerate}
\end{lemma}
\begin{proof}~
  \begin{enumerate}
  \item $\PF A$ is a strict submultiset of $\PF{A\wedge B}$.
  \item Let $\PF B=[C_i\Rightarrow\tau]_{i=1}^n$. Then, $\PF{A\Rightarrow
      B}=[(A\wedge C_i)\Rightarrow\tau]_{i=1}^n$. Hence, $m(A\Rightarrow
    B)\geq m(A\wedge C_1)+1>m(A\wedge C_1)\geq m(A)$.
  \item $m(A\Rightarrow B) = \sum_i m(A \wedge C_i) + 1 > \sum_i m(C_i) + 1 = m(B)$.
    \qedhere
  \end{enumerate}
\end{proof}

\subsection{Decomposition properties on types}\label{sec:DPT}
In simply typed lambda calculus, the implication and the conjunction are
constructors, that is $A\Rightarrow B$ is never equal to $C\wedge D$, if
$A\Rightarrow B=A'\Rightarrow B'$, then $A=A'$ and $B=B'$, and the same holds
for the conjunction. This is not the case in \OC, where $\tau\Rightarrow (\tau\wedge\tau)\equiv
(\tau\Rightarrow\tau)\wedge(\tau\Rightarrow\tau)$, but the connectors still have
some coherence properties:
\begin{itemize}
 \item If $A\Rightarrow B\equiv\bigwedge_{i=1}^n C_i$, then each $C_i$ is
   equivalent to an implication $A\Rightarrow B_i$, where the conjunction of the
   $B_i$ is equivalent to $B$.
  \item If $A\wedge B\equiv\bigwedge_iC_i$, then each $C_i$ is a conjunction of
    elements, possibly empty, that contribute to $A$ and to $B$.
\end{itemize}
We state these properties in Corollary~\ref{cor:ImpConj} and
Lemma~\ref{lem:eqConjN}. 

\begin{lemma}
  \label{lem:ImpConj}
  If $A\Rightarrow B\equiv C_1\wedge C_2$, then $C_1\equiv A\Rightarrow B_1$
  and $C_2\equiv A\Rightarrow B_2$ where
  $B\equiv B_1\wedge B_2$.
\end{lemma}
\begin{proof}
  By Lemma~\ref{lem:eqPF}, $\PF{A\Rightarrow B}\sim\PF{C_1\wedge
    C_2}=\PF{C_1}\uplus\PF{C_2}$.
  Let $\PF B=[D_i\Rightarrow\tau]_{i=1}^n$, so $\PF{A\Rightarrow B}=[(
  A\wedge D_i)\Rightarrow\tau]_{i=1}^n$.
  Without lost of generality, take $\PF{C_1}\sim[(A\wedge
  D_i)\Rightarrow\tau]_{i=1}^k$ and $\PF{C_2}\sim[(A\wedge
  D_i)\Rightarrow\tau]_{i=k+1}^n$. Therefore, by Lemma~\ref{lem:eqConjPF}, we have
  $A\Rightarrow
  B\equiv\bigwedge_{i=1}^k((A\wedge{D_i})\Rightarrow\tau)\wedge\bigwedge_{i=k+1}^n((A\wedge{D_i})\Rightarrow\tau)\equiv
  (A\Rightarrow\bigwedge_{i=1}^k (D_i\Rightarrow\tau))\wedge(A\Rightarrow\bigwedge_{i=k+1}^n (D_i\Rightarrow\tau))$.
  Take $B_1=\bigwedge_{i=1}^k{D_i}\Rightarrow\tau$ and $B_2=\bigwedge_{i=k+1}^n{D_i}\Rightarrow\tau$.
  Remark that $C_1 \equiv A\Rightarrow B_1$, $C_2\equiv A\Rightarrow B_2$ and
  $B\equiv B_1\wedge B_2$.
\end{proof}

\begin{corollary}\label{cor:ImpConj}
  If $A\Rightarrow B\equiv\bigwedge_{i=1}^n C_i$, then for $i\in\{1,\dots,n\}$, we have $C_i\equiv A\Rightarrow B_i$
  where $B\equiv\bigwedge_{i=1}^nB_i$.
\end{corollary}
\begin{proof}
  By induction on $n$.
  By Lemma~\ref{lem:ImpConj}, $\bigwedge_{i=1}^{n-1}C_i\equiv A\Rightarrow B'$ and
  $C_n\equiv A\Rightarrow B_n$, with $B\equiv B'\wedge B_n$. By
  the induction hypothesis, for $i\leq n-1$, $C_i\equiv A\Rightarrow B_i$ where
  $B'\equiv\bigwedge_{i=1}^{n-1}B_i$. Hence, $B\equiv\bigwedge_{i=1}^nB_i$.
\end{proof}

\begin{lemma}\label{lem:interMultiset}
  Let $R,S,T$ and $U$ be four multisets such that $R\uplus S=T\uplus U$, then
  there exist four multisets $V$, $W$, $X$, and $Y$ such that
  $R = V\uplus X$,
  $S = W\uplus Y$,
  $T = V\uplus W$,
  and $U = X\uplus Y$,
  cf.~Figure~\ref{fig:lemma}.
\end{lemma}
\begin{proof}
  Consider an element $a\in R\uplus S=T\uplus U$. Let $r$ be the multiplicity of
  $a$ in $R$, $s$ its multiplicity in $S$, $t$ its multiplicity in $T$, and $u$ its
  multiplicity in $U$. We have $r+s=t+u$.
  If $r\leq t$ we put $r$ copies of $a$ in $V$, $t-r$ in $W$, $0$ in $X$, and $u$
  in $Y$. Otherwise, we put $t$ in $V$, $0$ in $W$, $r-t$ in $X$, and $s$ in $Y$.
\end{proof}

\begin{figure}
  \begin{center}
  \begin{tikzpicture}
    \draw (-1,0) rectangle (-.5,.5) node[pos=0.5]{$W$};
    \draw (-.45,0) rectangle (.05,.5) node[pos=0.5]{$Y$};
    \draw (-1,0.55) rectangle (-.5,1.05) node[pos=0.5]{$V$};
    \draw (-0.45,0.55) rectangle (.05,1.05) node[pos=0.5]{$X$};
    \node at (-1.3,0.8) {$R$};
    \node at (-1.3,0.25) {$S$};
    \node at (-.75,1.3) {$T$};
    \node at (-.2,1.3) {$U$};
    \draw (.8,-.2) -- (.8,1.3);
    \draw (1.5,0) rectangle (2,.5) node[pos=0.5]{$W_1$};
    \draw (2.05,0) rectangle (2.55,.5) node[pos=0.5]{$W_2$};
    \draw (1.5,0.55) rectangle (2,1.05) node[pos=0.5]{$V_1$};
    \draw (2.05,0.55) rectangle (2.55,1.05) node[pos=0.5]{$V_2$};
    \draw (4.05,0) rectangle (4.55,.5) node[pos=0.5]{$W_n$};
    \draw (4.05,0.55) rectangle (4.55,1.05) node[pos=0.5]{$V_n$};
    \node at (1.2,0.8) {$R$};
    \node at (1.2,0.25) {$S$};
    \node at (1.75,1.3) {$T_1$};
    \node at (2.3,1.3) {$T_2$};
    \node at (4.3,1.3) {$T_n$};
    \draw[dotted] (2.8,0.8) -- (3.8,.8);
    \draw[dotted] (2.8,0.25) -- (3.8,.25);
  \end{tikzpicture}
\end{center}
\caption{}
  \label{fig:lemma}
\end{figure}

\begin{corollary}\label{cor:interMultiset}
  Let $R$ and $S$ be two multisets and $(T_i)_{i=1}^n$ be a family of multisets, such that $R\uplus S=\biguplus_{i=1}^nT_i$. Then, there exist multisets
  $V_1,\dots,V_n,W_1,\dots,W_n$ such that
  $R=\biguplus_i V_i$ and $S=\biguplus_i W_i$ and for each $i$, $T_i=V_i\uplus
  W_i$, cf.~Figure~\ref{fig:lemma}.
\end{corollary}
\begin{proof}
  By induction on $n$.
  We have $R\uplus S=\biguplus_{i=1}^{n-1}T_i\uplus T_n$. Then, by
  Lemma~\ref{lem:interMultiset}, there exist $R', S', V_n, W_n$ such that
  $R = R'\uplus V_n$,
  $S = S'\uplus W_n$,
  $\biguplus_{i=1}^{n-1}T_i = R'\uplus S'$,
  and
  $T_n = V_n\uplus W_n$.
  By induction hypothesis, there exist
  $V_1,\dots, V_{n-1}$ and $W_1,\dots, W_{n-1}$ such that
  $R'=\biguplus_{i=1}^{n-1}V_i$, $S'=\biguplus_{i=1}^{n-1}W_i$ and each $T_i=V_i\uplus W_i$. Hence, $R=\biguplus_{i=1}^n V_i$ and $S=\biguplus_{i=1}^n W_i$.
\end{proof}

\begin{lemma}
  \label{lem:eqConjN}
  If $A\wedge B\equiv\bigwedge_{i=1}^nC_i$ then there exists a partition
  $E\uplus F\uplus G$ of $\{1,\dots,n\}$ such that 
  \begin{itemize}
    \item $C_i = A_i\wedge B_i$, when $i\in E$;
    \item $C_i=A_i$, when $i\in F$;
    \item $C_i=B_i$, when $i\in G$;
    \item $A = \bigwedge_{i\in E\uplus F}A_i$; and
    \item $B = \bigwedge_{i\in E\uplus G}B_i$.
      \qed
  \end{itemize}
\end{lemma}
\begin{proof}
  Let $R = \PF{A}$, $S = \PF{B}$, and $T_i = \PF{C_i}$. By Lemma~\ref{lem:eqPF},
we have $\PF{A \wedge B}\sim\PF{\bigwedge_iC_i}$, that is $R\uplus
S\sim\biguplus_iT_i$. By Corollary~\ref{cor:interMultiset}, there exist $V_i$
and $W_i$ such that $R=\biguplus_{i=1}^n V_i$, $S=\biguplus_{i=1}^n
W_i$, and $T_i \sim V_i\uplus W_i$. As $T_i$ is non-empty, $V_i$ and $W_i$
cannot be both empty.
\begin{itemize}
\item If $V_i$ and $W_i$ are both non-empty, we let $i\in E$ and
  $A_i=\conj{V_i}$ and $B_i=\conj{W_i}$.
  By Lemma~\ref{lem:eqConjPF}, $C_i\equiv \conj{T_i}\equiv\conj{V_i\uplus W_i}\equiv
  A_i\wedge B_i$.
\item If $V_i$ is non-empty and $W_i$ is empty, we let $i\in F$, and
  $A_i=\conj{V_i}\equiv\conj{T_i} \equiv C_i$.
\item If $W_i$ is non-empty and $V_i$ is empty, we let $i\in G$, and
  $B_i=\conj{W_i}\equiv \conj{T_i} \equiv C_i$.
\end{itemize}

As $V_i=\emptyset$ when $i\in G$, we have $A\equiv\conj R\equiv\conj{\biguplus_{i\in E\uplus F}V_i}\equiv\bigwedge_{i\in E\uplus F}A_i$.

As $W_i=\emptyset$ when $i\in F$, we have $B\equiv\conj S\equiv\conj{\biguplus_{i\in E\uplus
    G}W_i}\equiv\bigwedge_{i\in E\uplus G}B_i$.
\end{proof}

\section{\OC}\label{sec:calculus}
\subsection{Syntax}
We associate to each (up to equivalence) prime type $A$ an infinite set of variables $\V_A$ such that if
$A\equiv B$ then $\V_A=\V_B$ and if $A\not\equiv B$ then
$\V_A\cap\V_B=\emptyset$.
The set of terms is defined inductively by the grammar
\[
  \ve r,\ve s,\ve t,\dots\ ::=\quad x~|~\lambda x.\ve r~|~\ve r\ve s~|~\ve
  r\times\ve s~|~\pi_A(\ve r)
\]
We recall the type on binding occurrences of variables and write $\lambda x^A.t$
for $\lambda x.t$ when $x\in\V_A$. $\alpha$-equivalence and substitution are
defined as usual. The type system is given in Table~\ref{tab:typeSys}. We use a
presentation of typing rules without explicit context
following~\cite{GeuversKrebbersMcKinnaWiedijkLFMTP10,ParkSeoParkLeeJAR13}, hence
the typing judgments have the form $\ve r:A$. The preconditions of a
typing rule is written on its left.

\begin{table}
      \[
        \condi{\cond{x\in\V_A}}{\infer[^{(ax)}]{x:A}{\phantom{x:A}}}
        \qquad
        \condi{\cond{A\equiv B}}{\infer[^{(\equiv)}]{\ve r:B}{\ve r:A}}
      \]
      \[
        \infer[^{(\Rightarrow_i)}]{\lambda x^A.\ve r:A\Rightarrow
          B}{\ve r:B}
        \qquad
        {\infer[^{(\Rightarrow_e)}]{\ve r\ve s:B}{\ve r:A\Rightarrow B & \ve s:A}}
        \qquad
        \infer[^{(\wedge_i)}]{\ve r\times\ve s:A\wedge B}{\ve r:A & \ve s:B}
        \qquad
        \infer[^{(\wedge_e)}]{\pi_A(\ve r):A}{\ve r:A\wedge B}
      \]
      \caption{The type system.}
  \label{tab:typeSys}
\end{table}

\subsection{Operational semantics}\label{sec:opSem}
The operational semantics of the calculus is defined by two relations: an
equivalence relation, and a reduction relation.

\begin{definition}
The symmetric relation $\eq$ is the smallest contextually closed relation
defined by the rules given in
Table~\ref{tab:opSemSym}.
\end{definition}
Each isomorphism induces an equivalence between terms. Two rules however
correspond to the isomorphism~\eqref{iso:distrib}, depending on which distribution
is taken into account: elimination or introduction of implication.
We write $\eq^*$ for the transitive and reflexive closure of $\eq$. Note that
$\eq^*$ is an equivalence relation.

\begin{table}
        \begin{align*}
          \ve r\times \ve s &\eq \ve s\times \ve r & \rulelabel{(comm)}\\
          (\ve r\times \ve s)\times \ve t &\eq \ve r\times (\ve s\times \ve t) & \rulelabel{(asso)}\\
          \lambda x^A.(\ve r\times \ve s) &\eq \lambda x^A.\ve r\times \lambda x^A.\ve s & \rulelabel{(dist$_{\lambda}$)}\\
          (\ve r\times \ve s)\ve t &\eq \ve r\ve t\times \ve s\ve t & \rulelabel{(dist$_{\mathsf{app}}$)}\\
          \ve r\ve s\ve t &\eq \ve r(\ve s\times \ve t) & \rulelabel{(curry)}
        \end{align*}
      \caption{Symmetric relation.}
  \label{tab:opSemSym}
\end{table}

Because of the associativity property of $\times$, the term $\ve r\times (\ve
s\times \ve t)$ is equivalent to the term $(\ve r\times \ve s)\times \ve t$, so we
can just write it $\ve r\times \ve s\times \ve t$. 

As explained in the introduction, variables of conjunctive types are useless,
hence all variables have prime types. This way, there is no term 
$\lambda x^{\tau\wedge\tau}.x$, but a term $\lambda y^\tau.\lambda
z^\tau.y\times z$ which is equivalent to $(\lambda y^\tau.\lambda
z^\tau.y)\times(\lambda y^\tau.\lambda z^\tau.z)$.

The size of a term is not invariant through the equivalence $\eq$. Hence, we
introduce a measure $M(\cdot)$, which is given in Table~\ref{tab:PandM}.
\begin{table}
  \[
    \begin{array}{rl|rl}
  P(x) &= 0& 
  M(x) &= 1\\ 
  P(\lambda x^A.\ve r) &= P(\ve r)& 
  M(\lambda x^A.\ve r) &= 1 + M(\ve r) + P(\ve r)\\ 
  P(\ve r\ve s) &= P(\ve r)& 
  M(\ve r\ve s) &= M(\ve r) + M(\ve s) + P(\ve r) M(\ve s)\\ 
  P(\ve r \times \ve s) &= 1 + P(\ve r) + P(\ve s)& 
  M(\ve r \times \ve s) &= M(\ve r) + M(\ve s)\\ 
  P(\pi_A(\ve r)) &= P(\ve r) &
  M(\pi_A(\ve r)) &= 1 +  M(\ve r) + P(\ve r)
    \end{array}
  \]
\caption{Measure on terms.}
\label{tab:PandM}
\end{table}

\begin{lemma}
  \label{lem:eqPr}
  If $\ve r\eq\ve s$ then $P(\ve r)=P(\ve s)$.
\end{lemma}
\begin{proof}
  We check the case of each rule of Table~\ref{tab:opSemSym}, and then conclude by structural induction to handle the contextual closure.
  \begin{itemize}
  \item \rulelabel{(comm)}: $P(\ve r\times\ve s)=1+P(\ve
    r)+P(\ve s)=P(\ve s\times\ve r)$.
  \item \rulelabel{(asso)}: $ P((\ve r\times\ve s)\times\ve t) =2+P(\ve r)+P(\ve s)+P(\ve t)
    =P(\ve r\times(\ve s\times\ve t)) $.

  \item \rulelabel{(dist$_\lambda$)}:
      $P(\lambda x^A.(\ve r\times\ve s))
      =1+P(\ve r)+P(\ve s)
      =P(\lambda x^A.\ve r\times\lambda x^A.\ve s)$.

    \item  \rulelabel{(dist$_{\mathsf{app}}$)}:
      $P((\ve r\times\ve s)\ve t)=1+P(r)+P(s)
      =P(\ve r\ve t\times\ve s\ve t)$.

  \item \rulelabel{(curry)}:  $P((\ve r\ve s)\ve
    t) =P(\ve r) =P(\ve r(\ve s\times\ve t))$.
    \qedhere
  \end{itemize}
\end{proof}

\begin{lemma}
  \label{lem:eqM}
  If $\ve r \eq\ve s$ then $M(\ve r) = M(\ve s)$.
\end{lemma}
\begin{proof}
  We check the case of each rule of Table~\ref{tab:opSemSym}, and then conclude by structural induction to handle the contextual closure.
  \begin{itemize}
  \item \rulelabel{(comm)}:
    \( M(\ve r \times\ve s) =
    M(\ve r) + M(\ve s) = M(\ve s \times\ve r). \)

  \item \rulelabel{(asso)}:
    \( M((\ve r \times\ve s) \times\ve t) = M(\ve r) + M(\ve s) + M(\ve t) =
    M(\ve r \times (\ve s \times\ve t)). \)

  \item \rulelabel{(dist$_{\lambda}$)}:
    \(
      M(\lambda x^A . (\ve r \times\ve s))
       = 2 + M(\ve r) + M(\ve s) + P(\ve r) + P(\ve s)
       = M(\lambda x^A . \ve r \times \lambda x^A . \ve s)
    \)

  \item \rulelabel{(dist$_{\mathsf{app}}$)}:
    $
      M((\ve r \times\ve s) \ve t)
      = M(\ve r) + M(\ve s) + 2 M(\ve t) + P(\ve r) M(\ve t) + P(\ve s) M(\ve t)
      = M(\ve r \ve t \times\ve s \ve t)
    $

  \item \rulelabel{(curry)}:
    \(
      M((\ve r \ve s) \ve t)
      = M(\ve r)+M(\ve s)+P(\ve r)M(\ve s)+M(\ve t)+P(\ve r)M(\ve t)
      = M(\ve r(\ve s \times\ve t))
    \)
    \qedhere
  \end{itemize}
\end{proof}

\begin{lemma}\label{lem:newlemma}
  $M(\lambda x^A.r)> M(r)$,
  $M(rs) > M(r)$,
  $M(rs) > M(s)$,
  $M(r \times s) >  M(r)$,
  $M(r \times s) > M(s)$, and
  $M(\pi_A(r)) > M(r)$.
\end{lemma}
\begin{proof}
By induction on $r$, $M(r)\geq 1$. We conclude with a case inspection.
\end{proof}

We use the measure to prove that the equivalence class of a term is a finite set.

\begin{lemma}
  \label{lem:finiteClasses}
  For any term $\ve r$, the set $\{\ve s~|~\ve s \eq^*\ve r\}$ is finite (modulo
  $\alpha$-equivalence).
\end{lemma}
\begin{proof}
  Since $\{\ve s~|~\ve s \eq^*\ve r\} \subseteq \{\ve s~|~FV(\ve s) = FV(\ve
r)~\mbox{and}~M(\ve s) = M(\ve r)\} \subseteq \{\ve s~|~FV(\ve s) \subseteq
FV(\ve r)~\mbox{and}~M(\ve s) \leq M(\ve r)\}$, where $FV(t)$ is the set of free
variables of $t$, all we need to prove is that for all natural numbers $n$, for
all finite sets of variables $F$, the set $H(n,F) = \{\ve s~|~FV(\ve s)
\subseteq F~\mbox{and}~M(\ve s) \leq n\}$ is finite.

  By induction on $n$. For $n = 1$ the set $\{\ve s~|~FV(\ve s) \subseteq
F~\mbox{and}~M(\ve s) \leq 1\}$ contains only the variables of $F$. Assume the
property holds for $n$, then, by the Lemma~\ref{lem:newlemma} the set $H(n+1,F)$
is a subset of the finite set containing the variables of $F$, the abstractions
$(\lambda x^A .\ve r)$ for $\ve r$ in $H(n,F \cup \{x\})$, the applications
$(\ve r\ve s)$ for $\ve r$ and $\ve s$ in $H(n,F)$, the products $\ve
r\times\ve s$ for $\ve r$ and $\ve s$ in $H(n,F)$, the projections $\pi_A(\ve
r)$ for $\ve r$ in $H(n,F)$.
\end{proof}

\begin{definition}
The reduction relation $\re$ is the smallest contextually closed relation
defined by the rules given in Table~\ref{tab:opSem}.
We write $\re^*$ for the transitive and reflexive closure of $\re$. 
\end{definition}
\begin{definition}
We write
$\toreq$ for the relation $\re$ modulo $\eq^*$ (i.e. $\ve r\toreq\ve s$ iff $\ve
r\eq^*\ve r'\re\ve s'\eq^*\ve s$), and $\toreq^*$ for its transitive and reflexive closure.
\end{definition}
Remark that, by Lemma~\ref{lem:finiteClasses}, a term has a finite number of
reducts in one step and these reducts can be computed.
\begin{table}
  \[
    \begin{array}{r@{\ }l@{\qquad}r@{\hspace{3cm}}r@{\ }l@{\qquad}r}
      \mbox{If }\ve s:A,\  (\lambda x^A.\ve r)\ve s &\re \ve r[\ve s/x] &
                                                                          \rulelabel{($\beta$)} &   \mbox{If }\ve r:A,\ \pi_A(\ve r\times\ve s) &\re \ve r & \rulelabel{($\pi$)}
    \end{array}
  \]
	\caption{Reduction relation.}
  \label{tab:opSem}
\end{table}

\subsection{Examples}

\begin{example}
  Let $\ve r:A$ and $\ve s:B$. Then $(\lambda x^{A}.\lambda
  y^B.x)(\ve r\times \ve s):A$ and
  \[
    (\lambda x^{A}.\lambda y^B.x)(\ve r\times \ve s) \eq (\lambda x^{A}.\lambda
    y^B.x)\ve r\ve s\re^* \ve r
  \]
  However, if $A\equiv B$, it is also possible to reduce in the following way
  \[
    (\lambda x^{A}.\lambda
    y^A.x)(\ve r\times \ve s) \eq (\lambda x^{A}.\lambda y^A.x)(\ve s\times \ve
    r) \eq (\lambda x^{A}.\lambda y^A.x)\ve s\ve r \re^* \ve s
  \]
  Hence, the usual encoding of the projector also behaves non-deterministically. 
\end{example}

\begin{example} Let $s:A$ and $\ve t:B$, and let $\tf=\lambda
  x^A.\lambda y^B.(x\times y)$.

  Then
  $\tf:A\Rightarrow B\Rightarrow (A\wedge B)\equiv ((A\wedge B)\Rightarrow
  A)\wedge((A\wedge B)\Rightarrow B)$. Therefore,
  $\pi_{(A\wedge B)\Rightarrow A}(\tf):(A\wedge B)\Rightarrow A$.
  Hence, $\pi_{(A\wedge B)\Rightarrow A}(\tf)(s\times t):A$.
  
  This term reduces as follows:
  \begin{align*}
    \pi_{(A\wedge B)\Rightarrow A}(\tf)(\ve s\times\ve t)
    &\eq\pi_{(A\wedge B)\Rightarrow A}(\tf)\ve s\ve t\\
    &\eq\pi_{(A\wedge B)\Rightarrow A}(\lambda x^A.(\lambda y^B.x)\times(\lambda y^B.y))st\\
    &\eq\pi_{(A\wedge B)\Rightarrow A}((\lambda x^A.\lambda y^B.x)\times(\lambda x^A.\lambda y^B.y))st\\
    &\re(\lambda x^A.\lambda y^B.x)st\\
    &\re(\lambda y^B.s)t
    \re \ve s
  \end{align*}
\end{example}

\begin{example}\label{ex:bool}
  Let $\true=\lambda x^A.\lambda y^B.x$ and $\false=\lambda x^A.\lambda y^B.y$.
  The term
  $\true\times \false\times \tf$ has type
  $((A\wedge B)\Rightarrow(A\wedge B))\wedge((A\wedge B)\Rightarrow(A\wedge B))$.
      
  Hence, $\pi_{(A\wedge B)\Rightarrow(A\wedge B)}(\true\times \false\times \tf)$ is well typed and reduces
non-deterministically either to $\true\times \false$ or to $\tf$. Moreover, as
$\true\times \false$ and $\tf$ are equivalent, the non-deterministic choice does
not play any role in this particular case. We will come back to the encoding of
booleans in \OC on Section~\ref{sec:bool}.
\end{example}

\section{Subject Reduction}\label{sec:SR}
The set of types assigned to a
term is preserved under $\eq$ and $\re$. Before proving this property, we 
prove the unicity of types (Lemma~\ref{lem:unicity}), the generation lemma
(Lemma~\ref{lem:generation}), and the substitution lemma
(Lemma~\ref{lem:substitution}). We only state the lemmas in this section. The
detailed proofs can be found in Appendix~\ref{app:SR}.

The following lemma states that a term can be typed only by equivalent types.

\begin{lemma}[Unicity]\label{lem:unicity}
  If $\ve r:A$ and $\ve r:B$, then $A\equiv B$.
\end{lemma}
\begin{proof}~
  \begin{itemize}
  \item If the last rule of the derivation of $\ve r:A$ is
    $(\equiv)$, then we have a shorter derivation of $r:C$ with
    $C\equiv A$, and, by the induction hypothesis, $C\equiv B$, hence $A\equiv
    B$.
  \item If the last rule of the derivation of $\ve r:B$ is $(\equiv)$ we proceed in the same way.
  \item All the remaining cases are syntax directed.
\qedhere
  \end{itemize}
\end{proof}

\begin{lemma}[Generation]
  \label{lem:generation}
  ~
  \begin{enumerate}
  \item\label{case:var} If $x\in\V_A$ and $x:B$, then $A\equiv B$.
  \item\label{case:lambda} If $\lambda x^A.\ve r:B$, then $B\equiv
    A\Rightarrow C$ and $\ve r:C$.
  \item If $\ve r\ve s:B$, then $\ve r:A\Rightarrow B$ and $\ve s:A$.
  \item If $\ve r\times \ve s:A$, then $A\equiv B\wedge C$ with $\ve r:B$ and $\ve s:C$.
  \item If $\pi_A(\ve r):B$, then $A\equiv B$ and $\ve r:B\wedge C$.
  \end{enumerate}
\end{lemma}
\begin{proof}
  Each statement is proved by induction on the typing derivation. For the
  statement~\ref{case:var}, we have $x\in\V_A$ and $x:B$. The only way to type
  this term is either by the rule $(ax)$ or $(\equiv)$.
  \begin{itemize}
  \item In the first case, $A=B$, hence $A\equiv B$.
  \item In the second case, there exists $B'$ such that $x:B'$ has
    a shorter derivation, and $B\equiv B'$. By the induction hypothesis
    $A\equiv B'\equiv B$.
  \end{itemize}
  For the statement~\ref{case:lambda}, we have $\lambda x^A.\ve
  r:B$. The only way to type this term is either by rule $(\Rightarrow_i)$,
  $(\equiv)$.
  \begin{itemize}
  \item In the first case, we have $B=A\Rightarrow C$ for some, $C$ and
    $\ve r:C$.
  \item In the second, there exists $B'$ such that $\lambda x^A.\ve
    r:B'$ has a shorter derivation, and $B\equiv B'$. By the induction
    hypothesis, $B'\equiv A\Rightarrow C$ and $\ve r:C$. Thus,
    $B\equiv B'\equiv A\Rightarrow C$.
  \end{itemize}
  The three other statements are similar.
\end{proof}

\begin{lemma}
  [Substitution]
  \label{lem:substitution}
  If $\ve r:A$, $\ve s:B$, and $x\in\V_B$, then $\ve r[\ve s/x]:A$.
\end{lemma}
\begin{proof}
  By structural induction on $r$ (cf.~Appendix~\ref{app:SR}).
\end{proof}

\begin{theorem}[Subject reduction]\label{thm:SR}
  If $\ve r:A$ and $\ve r\re\ve s$ or $\ve r\eq\ve s$ then $\ve s:A$.
\end{theorem}
\begin{proof}
  By induction on the rewrite relation (cf.~Appendix~\ref{app:SR}).
\end{proof}

\section{Strong Normalisation}\label{sec:SN}
In this section we prove the strong normalisation of reduction $\toreq$:
every
reduction sequence fired from a typed term eventually terminates.
The set of strongly normalising terms with respect to reduction~\toreq\ is
written $\SN$.
The size of the longest reduction issued from $t$ is written $|t|$ (recall that
each term has a finite number of reducts).

To prove that every term is in $\SN$, we associate, as usual, a set $\interp A$
of strongly normalising terms to each type $A$. A term $r:A$ is said to be
reducible when $r\in\interp A$. We then prove an adequacy theorem stating that
every well typed term is reducible.

In simply typed lambda calculus we can either define $\interp{A_1\Rightarrow
A_2\Rightarrow \cdots\Rightarrow A_n\Rightarrow\tau}$ as the set of terms $r$
such that for all $s\in\interp{A_1}$,
$rs\in\interp{A_2\Rightarrow\cdots\Rightarrow A_n\Rightarrow\tau}$ or,
equivalently, as the set of terms $r$ such that for all $\ve
s_i\in\interp{A_i}$, $rs_1\dots s_n\in\interp\tau=\SN$. To prove that a term of
the form $\lambda x^A.t$ is reducible, we need to use the so-called CR3
property~\cite{Girard89}, in the first case, and the property that a term whose
all one-step reducts are in $\SN$ is in $\SN$, in the second. In \OC, an
introduction can be equivalent to an elimination e.g.~$rt\times st\eq (r\times
s)t$, hence, we cannot define a notion of neutral term and have an equivalent to
the CR3 property. Therefore, we use the second definition.


Before we prove the normalisation of \OC, we first reformulate the proof of
strong normalisation of simply typed lambda-calculus along these lines.

\subsection{Normalisation of simply typed lambda calculus}\label{sec:SNinST}
\begin{definition}[Elimination context] 
  Consider an extension of simply typed lambda calculus where we introduce an
  extra symbol $\hole A$, called hole of type $A$.

  An elimination context with a hole $\hole{B_1\Rightarrow\dots\Rightarrow B_n\Rightarrow\tau}$
is a term $K^\tau_{B_1\Rightarrow\dots\Rightarrow B_n\Rightarrow\tau}$ of type $\tau$ of the form 
$\hole{B_1\Rightarrow\dots\Rightarrow B_n\Rightarrow\tau}r_1\dots r_n$. We write
$K^\tau_A[t]$, for the term $K^\tau_A[t/\hole A] = tr_1\dots r_n$. 
\end{definition}

\begin{definition}
  [Terms occurring in an elimination context]
$\T(\hole Ar_1\dots r_n) = 
\{r_1,\dots, r_n\}$. Note that the types of the elements of $\T(\hole Ar_1\dots r_n)$ are
smaller than $A$, and that if $r_1,\dots, r_n\in\SN$, then $\hole Ar_1\dots r_n\in\SN$. 
\end{definition}

\begin{definition}[Reducibility]
The set $\interp{A}$ of reducible terms of type $A$ is defined by
structural induction on $A$ as the set of terms $t:A$ such that for any
elimination context $K^\tau_A$ such that the terms in $\T(K^\tau_A)$ are
all reducible, we have $K^\tau_A[t]\in\SN$. 
\end{definition}

\begin{definition}
  [Reducible elimination context]
  An elimination context $K^B_A$ is reducible, if all the terms in $\T(K^B_A)$ are reducible.
\end{definition}

\begin{lemma}
  \label{lem:CR1ST}
  For all $A$, $\interp A\subseteq\SN$ and all the variables of type $A$ are in $\interp A$.
\end{lemma}
\begin{proof}
  By induction on $A$.
\end{proof}

\begin{lemma}[Adequacy of application]\label{app}
If $r \in \interp{A \Rightarrow B}$ and $s \in \interp{A}$, then $r s \in
\interp{B}$.
\end{lemma}
\begin{proof}
  Let $K^\tau_B$ be a reducible elimination context.
  We need to prove that $K^\tau_B[rs]\in\SN$. As $s\in\interp A$, the
  elimination context $K'^\tau_{A\Rightarrow B}=K^\tau_B[\hole{A\Rightarrow B}s]$ is reducible, and since
  $r\in\interp{A\Rightarrow B}$, we have $K^\tau_B[rs]=K'^\tau_{A\Rightarrow B}[r]\in\SN$.
\end{proof}

\begin{lemma}[Adequacy of abstraction]\label{lambda}
If for all $t \in \interp{A}$, $r[t/x] \in \interp{B}$, then 
$\lambda x^A.r \in \interp{A \Rightarrow B}$. 
\end{lemma}
\begin{proof}
We need to prove that for every reducible elimination context $K^\tau_{A\Rightarrow B}$, 
$K^\tau_{A\Rightarrow B}[\lambda x^A.r]\in\SN$, that is that all its one step
reducts are in $\SN$.
By Lemma~\ref{lem:CR1ST}, $x\in\interp A$, so $r\in\interp B\subseteq\SN$.
Then, we proceed by induction on $|r| + |K^\tau_{A\Rightarrow B}|$.
\end{proof}

\begin{definition}[Adequate substitution]
  A substitution $\sigma$ is adequate if for all $x:A$, we have
  $\sigma(x)\in\interp{A}$.
\end{definition}

\begin{theorem}[Adequacy]\label{thm:adequacyST}
  If $r:A$, the for all $\sigma$ adequate, we have $\sigma r\in\interp{A}$.
\end{theorem}
\begin{proof}
  By induction on $r$. 
\end{proof}

\begin{theorem}[Strong normalisation]
 If $r:A$, then $r\in\SN$.
\end{theorem}
\begin{proof}
By Lemma~\ref{lem:CR1ST}, the idendity substitution is adequate. Thus, 
 by Theorem~\ref{thm:adequacyST} and Lemma~\ref{lem:CR1ST}, 
 $r\in\interp{A}\subseteq\SN$.
\end{proof}

\subsection{Reduction of a product}\label{sec:RP}
When simply-typed lambda-calculus is extended with pairs, proving that if
$r_1\in\SN$ and $r_2\in\SN$ then $r_1\times r_2\in\SN$ is easy. However, in \OC
this property (Lemma~\ref{lem:prodOfSN}) is harder to prove, as it requires a
characterisation of the terms equivalent to the product $r_1\times r_2$
(Lemma~\ref{lem:eqProd}) and of all the reducts of this term
(Lemma~\ref{lem:reProd}).

In Lemma~\ref{lem:eqProd}, we characterise the terms equivalent to a product.

\begin{lemma} \label{lem:eqProd} If $\ve r\times \ve s\eq^*\ve t$ then either
  \begin{enumerate}
  \item\label{it:eqProd-sum} $\ve t=\ve u\times \ve v$ where either
    \begin{enumerate}
    \item $\ve u\eq^*\ve t_{11}\times \ve t_{21}$ and $\ve v\eq^*\ve
      t_{12}\times \ve t_{22}$ with $\ve r\eq^*\ve t_{11}\times \ve t_{12}$ and
      $\ve s\eq^*\ve t_{21}\times \ve t_{22}$, or
    \item $\ve v\eq^*\ve w\times \ve s$ with $\ve r\eq^*\ve u\times \ve w$, or
      any of the three symmetric cases, or
    \item $\ve r\eq^*\ve u$ and $\ve s\eq^*\ve v$, or the symmetric case.
    \end{enumerate}
  \item\label{it:eqProd-lam} $\ve t=\lambda x^A.\ve a$ and $\ve a\eq^*\ve a_1\times \ve
    a_2$ with $\ve r\eq^*\lambda x^A.\ve a_1$ and $\ve s\eq^*\lambda x^A.\ve
    a_2$.
  \item\label{it:eqProd-app} $\ve t=\ve a\ve v$ and $\ve a\eq^*\ve a_1\times \ve a_2$,
    with $\ve r\eq^*\ve a_1\ve v$ and $\ve s\eq^*\ve a_2\ve v$.
  \end{enumerate}
\end{lemma}
\begin{proof}
  By a double induction, first on $M(\ve t)$ and then on the length of the relation $\eq^*$ (cf.~Appendix~\ref{app:RP}).
\end{proof}

In Lemma~\ref{lem:reProd}, we characterise the reducts of a product.

\begin{lemma}
  \label{lem:reProd}
  If $\ve r_1\times \ve r_2\eq^*\ve s\re\ve t$, there exists $\ve u_1$, $\ve
  u_2$ such that $\ve t\eq^*\ve u_1\times \ve u_2$ and either ($\ve r_1\toreq\ve
  u_1$ and $\ve r_2\toreq\ve u_2$), or ($\ve r_1\toreq\ve u_1$ and $\ve
  r_2\eq^*\ve u_2$), or ($\ve r_1\eq^*\ve u_1$ and $\ve r_2\toreq\ve u_2$).
\end{lemma}
\begin{proof}
  By induction on $M(\ve r_1\times \ve r_2)$. 
\end{proof}

\begin{lemma}\label{lem:prodOfSN}
  If $\ve r_1\in\SN$ and $\ve r_2\in\SN$, then $\ve r_1\times \ve r_2\in\SN$.
\end{lemma}
\begin{proof}
  By Lemma~\ref{lem:reProd}, from a reduction sequence starting from $\ve
  r_1\times \ve r_2$ we can extract one starting from $\ve r_1$,
  or $\ve r_2$ or both. Hence, this reduction sequence is finite.
\end{proof}

\subsection{Reduction of a term of a conjunctive type}\label{sec:CT}
The next lemma takes advantage of the fact that all the variables have prime
types to prove that all terms of conjunctive type, even open ones, reduce to a
product. For instance, instead of the term $\lambda x^{\tau\wedge\tau}.x$, of
type $((\tau\wedge\tau)\Rightarrow\tau)\wedge((\tau\wedge\tau)\Rightarrow\tau)$,
we must write $\lambda y^\tau.\lambda z^\tau.y\times z$, which is equivalent to $(\lambda y^\tau.\lambda z^\tau.y)\times (\lambda y^\tau.\lambda z^\tau.z)$.
\begin{lemma}\label{lem:lelemme}
  If $\ve r:\bigwedge_{i=1}^n A_i$, then $\ve r\toreq^*\prod_{i=1}^n\ve
  r_i$ where $r_i:A_i$.
\end{lemma}
\begin{proof}
  By induction on $r$.
  \begin{itemize}
  \item $\ve r=x$, then it has a prime type, so take $r_1=r$.
  \item Let $\ve r=\lambda x^C.\ve s$. Then, by Lemma~\ref{lem:generation},
    $\ve s:D$ with $C\Rightarrow D\equiv \bigwedge_iA_i$. So, by
    Corollary~\ref{cor:ImpConj}, $D\equiv\bigwedge_i D_i$, and so, by the induction
    hypothesis, $\ve s\toreq^*\prod_i\ve s_i$. Therefore, $\lambda x^C.\ve
    s\toreq^*\prod_i\lambda x^C.\ve s_i$.
  \item Let $\ve r=\ve s\ve t$. Then, by Lemma~\ref{lem:generation},
    $s:C\Rightarrow \bigwedge_iA_i$, so $s:\bigwedge_i(C\Rightarrow
    A_i)$. Therefore, by the induction hypothesis, $\ve
    s\toreq^*\prod_i\ve s_i$, and so $\ve s\ve t\toreq^*\prod_i\ve s_i\ve t$. 
  \item Let $\ve r=\ve s\times\ve t$. Then, by
    Lemma~\ref{lem:generation}, $s:B$ and
    $t:C$, with $B\wedge C\equiv\bigwedge_iA_i$. By
    Lemma~\ref{lem:eqConjN}, 
    there exists a partition $E\uplus F\uplus G$ of $\{1,\dots,n\}$ such that
$A_i\equiv B_i\wedge C_i$, when $i\in E$;
$A_i\equiv B_i$, when $i\in F$;
$A_i\equiv C_i$, when $i\in G$;
$B\equiv \bigwedge_{i\in E\uplus F}B_i$; and
$C\equiv \bigwedge_{i\in E\uplus G}C_i$.
    By the induction hypothesis, $\ve s\toreq^*\prod_{i\in E\uplus F}s_i$ and
    $t\toreq^*\prod_{i\in E\uplus G} t_i$.
    If $i\in E$, we let $r_i = s_i\times t_i$, if $i\in F$, we let $r_i=s_i$, if
    $i\in G$, we let $r_i=t_i$. We have $r=s\times t\toreq^*\prod_{i=1}^n r_i$.

  \item Let $r=\pi_{\bigwedge_iA_i}(s)$. Then, by
    Lemma~\ref{lem:generation}, $s:\bigwedge_iA_i\wedge B$, and hence, by
    the induction hypothesis, $s\toreq^*\prod_is_i\times t$ where
    $s_i:A_i$ and $t:B$, hence $r\toreq\prod_is_i$.
    \qedhere
  \end{itemize}
\end{proof}

\begin{corollary}\label{cor:lelemme}
  If $r:A\wedge B$, then $r\toreq^* r_1\times r_2$ where $r_1:A$ and $r_2:B$.
\end{corollary}
\begin{proof}
  Let $\PF A=[A_i]_i$, $\PF B=[B_j]_j$, by Lemma~\ref{lem:eqConjPF},
  $A\wedge B\equiv\bigwedge_iA_i\wedge\bigwedge_jB_j$. Then, by
  Lemma~\ref{lem:lelemme}, $r\toreq^* \prod_ir_{1i}\times\prod_jr_{2j}$. Take
  $r_1=\prod_ir_{1i}$ and $r_2=\prod_jr_{2j}$.
\end{proof}

\subsection{Reducibility}\label{sec:Red}
\begin{definition}
  [Elimination context]
  Consider an extension of the language where we introduce an extra symbol
  $\hole A$, called hole of type $A$. We define the set of elimination contexts with a hole $\hole A$ as
  the smallest set such that:
  \begin{itemize}
  \item $\hole A$ is an elimination context of type $A$,
  \item if $K_A^{B\Rightarrow C}$ is an elimination context of type $B
    \Rightarrow C$ with a hole of type $A$, and $r:B$ then
    $K_A^{B\Rightarrow C}r$ is an elimination context of type $C$ with a hole
    of type $A$, 
\item and if $K_A^{B\wedge C}$ is an elimination context of type $B \wedge C$
  with a hole of type $A$, then
$\pi_B(K_A^{B\wedge C})$ is an elimination context of type $B$ with a hole of
type $A$.
\end{itemize}
We write $K_A^{B}[t]$ for $K_A^B[t/\hole A]$, where $\hole A$ is the hole of
$K_A^B$. In particular, $t$ may be an elimination context.
\end{definition}
\begin{example}
  Let $K^\tau_\tau=\hole\tau$ and
  $K'^\tau_{\tau\Rightarrow(\tau\wedge\tau)}=K^\tau_\tau[\pi_\tau(\hole{\tau\Rightarrow(\tau\wedge\tau)}x)]$.
  Then
  $K'^\tau_{\tau\Rightarrow(\tau\wedge\tau)}=\pi_\tau(\hole{\tau\Rightarrow(\tau\wedge\tau)}x)$,
  and $K'^\tau_{\tau\Rightarrow(\tau\wedge\tau)}[\lambda y^\tau.y\times y]=\pi_\tau((\lambda y^\tau.y\times y)x)$.
\end{example}

\begin{definition}
  [Terms occurring in an elimination context]
  Let $K_A^B$ be an elimination context.
  The multiset of terms occurring in $K_A^B$ is defined as
  \[
    \T(\hole A) = \emptyset;\qquad
    \T(K^{B\Rightarrow C}_Ar) = \T(K^{B\Rightarrow C}_A) \uplus \{r\};\qquad
    \T(\pi_B(K^{B\wedge C}_A)) =\T(K^{B\wedge C}_A)
  \]
  We write $|K^B_A|$ for $\sum_{i=1}^n|r_i|$ where $[r_1,\dots,r_n]=\T(K^B_A)$.
\end{definition}

\begin{example}
  $\T(\hole Ars) = [r,s]$ and $\T(\hole A(r\times s)) = [r\times s]$.
Remark that $K^B_A[t]\eq^* K'^{B}_{A}[t]$ does not imply $\T(K^B_A)\sim{\cal
  T}(K'^B_A)$.
\end{example}

Remark that if $K^B_A$ is a context, $m(B)\leq m(A)$ and hence if a term in
$\T(K^B_A)$ has type $C$, then $m(C)<m(A)$.
\begin{lemma}\label{lem:KvarSN}
  Let $K^\tau_A$ be an elimination context such that $\T(K^\tau_A)\subseteq\SN$,
  let $\PF A=[B_1,\dots,B_n]$, and let $x_i\in\V_{B_i}$.
  Then $K^\tau_A[x_1\times\dots\times x_n]\in\SN$.
\end{lemma}
\begin{proof}
By induction on the number of projections in $K^\tau_A$.
  \begin{itemize}
  \item If $K^\tau_A$ does not contain any projection, then it has the
    form $[]^A~r_1~...~r_m$. Let $C_i$ be the type of $r_i$, we have
    $A \equiv C_1 \Rightarrow ... \Rightarrow C_n \Rightarrow \tau$,
    thus $A$ is prime, $n=1$, and we need to prove that
    $x_1~r_1~...~r_m$ is in \SN\ which is a consequence of the fact that
    $r_1$, ..., $r_n$ are in \SN.

  \item Otherwise, $K^\tau_A = K'^\tau_{B}[\pi_{B}(\hole Ar_1\dots r_m)]$, and
    $K^\tau_A[x_1\times\dots\times x_n] =
    K'^\tau_{B}[\pi_{B}((x_1\times\dots\times x_n)r_1\dots r_m)] \eq^*
    K'^\tau_{B}[\pi_{B}(x_1r_1\dots r_m\times\dots\times x_nr_1\dots r_m)]$.  We
    prove that this term is in $\SN$ by showing, more generally, that if $s_i$
    are reducts of $x_ir_1\dots r_m$, then
    $K'^\tau_{B}[\pi_{B}(s_1\times\dots\times s_n)]\in\SN$. To do so, we show,
    by induction on $|K'^\tau_B|+|s_1\times\dots\times s_n|$, that all the one
    step reducts of this term are in $\SN$.
    \begin{itemize}
    \item If the reduction takes place in one of the terms in $\T(K'^\tau_B)$ or
      in one of the $s_i$, we apply the induction hypothesis.
    \item Otherwise, the reduction is a \rulelabel{($\pi$)} reduction of $\pi_B
      (s_1\times\dots\times s_n)$ yielding, without lost of generality, a term
      of the form $K'^\tau_B[s_1\times\dots\times s_q]$.
      This term is a reduct of
      $K'^\tau_{B}[(x_1\times\dots\times x_q)r_1\dots r_m]$.
      As the context
      $K'^\tau_{B}[([]^C~r_1\dots r_m)]$
      contains one projection less than
      $K^\tau_A$ this term is in $\SN$. Hence so does its reduct
      $K'^\tau_B[s_1\times\dots\times s_q]$.
      \qedhere
   \end{itemize}
  \end{itemize}
\end{proof}

\begin{definition}[Reducibility]\label{def:interpretation}
  The set $\interp A$ of reducible terms of type $A$ is
  defined by induction on $m(A)$ as the set of terms $t:A$ such that for any elimination context $K^\tau_A$ such that the terms of $\T(K^\tau_A)$ are all reducible, we have $K^\tau_A[t]\in\SN$.
\end{definition}

\begin{definition}
  [Reducible elimination context]
  An elimination context $K^B_A$ is reducible, if all the terms in ${\cal
    T}(K^B_A)$ are reducible.

  From now on we consider all the elimination contexts to be reducible.
\end{definition}

The following lemma is a trivial consequence of the definition of reducibility.
\begin{lemma}\label{lem:eqInterp}
  If $A\equiv B$, then $\interp A=\interp B$.
  \qed
\end{lemma}

\begin{lemma}
  \label{lem:CR1}
  For all $A$, $\interp A\subseteq\SN$ and $\interp A\neq\emptyset$.
\end{lemma}
\begin{proof}
  By induction on $m(A)$. By the induction hypothesis, for all the $B$ such that
$m(B)<m(A)$, $\interp B\neq\emptyset$. Thus, there exists an elimination context
$K^\tau_A$. Hence, if $r\in\interp A$, $K^\tau_A[r]\in\SN$, hence $r\in\SN$.

We then prove that if $\PF A=[B_1,\dots,B_n]$ and $x_i\in\V_{B_i}$ then $\prod_ix_i\in\interp A$.
By the induction hypothesis, $\T(K^\tau_A)\subseteq\SN$,
hence, by Lemma~\ref{lem:KvarSN}, $K^\tau_A[x_1\times\dots\times x_n]\in\SN$.
\end{proof}

\subsection{Adequacy}\label{sec:Adequacy}
We finally prove the adequacy theorem (Theorem~\ref{thm:adequacy}) showing that
every typed term is reducible,
 and the strong normalisation theorem (Theorem~\ref{thm:SN}) as a
consequence of it.

\begin{lemma}[Adequacy of projection]\label{lem:SNimpliespiSN}
  If $r\in\interp{A\wedge B}$, then $\pi_A(r)\in\interp A$.
\end{lemma}
\begin{proof}
  We need to prove that $K^\tau_A[\pi_A(r)]\in\SN$. Take $K'^\tau_{A\wedge B}=
  K^\tau_A[\pi_A\hole{A\wedge B}]$
  and since $r\in\interp{A\wedge B}$, we have $K^\tau_A[\pi_A(r)]=K'^\tau_{A\wedge B}[r]\in\SN$.
\end{proof}

\begin{lemma}
  [Adequacy of application]\label{lem:appOfInterp}
  If $r\in\interp{A\Rightarrow B}$, and
  $s\in\interp A$,
  then $rs\in\interp B$.
\end{lemma}
\begin{proof}
  We need to prove that $K^\tau_B[rs]\in\SN$. Take $K'^\tau_{A\Rightarrow
    B}=K^\tau_B[\hole{A\Rightarrow B}s]$ and since
  $r\in\interp{A\Rightarrow B}$, we have $K^\tau_B[rs]=K'^\tau_{A\Rightarrow B}[r]\in\SN$.
\end{proof}

\begin{lemma}[Adequacy of product]\label{lem:adequacyOfProd}
  If $r\in\interp A$ and $s\in\interp B$, then
  $\ve r\times \ve s\in\interp{A\wedge B}$.
\end{lemma}
\begin{proof}
  We need to prove that $K^\tau_{A\wedge B}[r\times s]\in\SN$.
  We proceed by induction on the number of projections in $K^\tau_{A\Rightarrow B}$.
  Since the hole of
$K^\tau_{A\wedge B}$ has type $A\wedge B$, and $K^\tau_{A\wedge B}[t]$ has type
$\tau$ for any $t:A$, we can assume, without lost of generality, that the context
$K^\tau_{A\wedge B}$ has the form $K'^\tau_C[\pi_C (\hole{A\wedge B}t_1\dots
t_n)]$.
  We prove that all $K'^\tau_C[\pi_C (rt_1\dots t_n\times st_1\dots t_n)]\in\SN$
  by showing, more generally, that if $r'$ and $s'$ are two reducts of
  $rt_1\dots t_n$ and $st_1\dots t_n$, then
  $K'^\tau_C[\pi_C(r'\times s')]\in\SN$. For this, we show that all its one step reducts are
  in $\SN$, by induction on $|K'^\tau_C| + |r'|+|s'|$.
  Full details are given in Appendix~\ref{app:Adequacy}.
\end{proof}

\begin{lemma}[Adequacy of abstraction]\label{lem:lamOfInterp}
  If for all $t\in\interp A$, $r[t/x]\in\interp B$, then
  $\lambda x^A.r\in\interp{A\Rightarrow B}$.
\end{lemma}
\begin{proof}
  By induction on $M(r)$. In the case $r\not\eq^*r_1\times r_2$, 
  we need to prove that for any elimination context $K^\tau_{A\Rightarrow B}$, we have
  $K^\tau_{A\Rightarrow B}[\lambda x^A. r]\in\SN$, and we do so by a second
  induction on $|K^\tau_{A\Rightarrow B}| + |r|$ to show that all
  the one step reducts of $K^\tau_{A\Rightarrow B}[\lambda x^A.r]$ are in $\SN$.
  Full details are given in Appendix~\ref{app:Adequacy}.
\end{proof}

\begin{definition}[Adequate substitution]
  A substitution $\sigma$ is adequate if for all $x\in\V_A$, we have
  $\sigma(x)\in\interp{A}$.
\end{definition}

\begin{theorem}[Adequacy]\label{thm:adequacy}
  If $\ve r:A$, then for all $\sigma$ adequate, we have $\sigma r\in\interp A$.
\end{theorem}
\begin{proof}
  By induction on $r$. 
  Full details are given in Appendix~\ref{app:Adequacy}.
\end{proof}

\begin{theorem}[Strong normalisation]\label{thm:SN}
  If $\ve r:A$, then $\ve r\in\SN$.
\end{theorem}
\begin{proof}
  By Lemma~\ref{lem:CR1}, the identity substitution is adequate. Thus,
  by Theorem~\ref{thm:adequacy} and Lemma~\ref{lem:CR1}, $r\in\interp{A}\subseteq\SN$.
\end{proof}

\section{Consistency}\label{sec:cons}

\begin{lemma}\label{lem:eqK}
  For any term $r:A$ there exists an elimination context $K^A_B$ and a term
  $s:B$, which is not an elimination, such that $r=K^A_B[s]$.
\end{lemma}
\begin{proof}
  We proceed by structural induction on $r$.
  \begin{itemize}
  \item If $r$ is a variable, an abstraction, or a product, we take $s=r$ and
  $K^A_A=\hole A$.
  \item If $r$ is an application $r_1r_2$, by the induction hypothesis,
    $r_1=K^{C\Rightarrow B}_A[s]$, we take $K'^B_A = K^{C\Rightarrow B}_Ar_2$.
  \item If $r$ is a projection $\pi_A(r')$, by the induction hypothesis,
    $r'=K^{A\wedge C}_B[s]$, we take $K'^A_B = \pi_A(K^{A\wedge C}_B)$.
  \qedhere
  \end{itemize}
\end{proof}

\begin{corollary}\label{cor:consistency}
  There is no closed normal term of type $\tau$.
\end{corollary}
\begin{proof}
  Let $r:\tau$ be a closed normal term.
  By Lemma~\ref{lem:eqK}, any $r=K^\tau_A[s]$,
  where $s$ is not an elimination.
  Since the term is closed, $s$ is not a variable. Thus it is either and
  abstraction or a product. 
  \begin{itemize}
  \item If $A$ is prime, then, $K^\tau_A$ cannot contain a projection, so by rule \rulelabel{(curry)}
   we have $K^\tau_A \eq^* \hole At$, with $t:B$, and $s$ has the form $\lambda
   x^C.s'$ with $s':D\Rightarrow\tau$. We have $B\equiv C\wedge D$.
    By Corollary~\ref{cor:lelemme}, and since $t$ is normal, $t\eq^*
   t_1\times t_2$ where $t_1:C$ and $t_2:D$, so
   $K^\tau_A\eq^* \hole At_1t_2$, hence,
   $r=K^\tau_A[\lambda x^C.s']\eq^*(\lambda
   x^C.s')t_1t_2$ is not normal.
  \item Otherwise, $K^\tau_A = K'^\tau_B[\pi_B(\hole A t_1\dots
    t_n)]$, with $\hole At_1\dots t_n:B\wedge C$.
    Then, by Corollary~\ref{cor:lelemme}, and since $st_1\dots t_n$ is normal, $st_1\dots t_n\eq^*s_1\times s_2$,
    thus $r=K^\tau_A[s]\eq^*K'^\tau_B[\pi_B(s_1\times s_2)]$, which is not normal.
    \qedhere
  \end{itemize}
\end{proof}

\section{Computing with \OC}\label{sec:computing}
\label{sec:pairs}
Because the symbol $\times$ is associative and commutative, \OC does not
contain the usual notion of pairs. However, it is possible to encode a
deterministic projection, even if we have more than one term of the same type.
An example, although there are various possibilities, is
to encode the pairs $\langle r,s \rangle:A\times A$ as
$\lambda x^{\numbertype{1}}.\ve r \times  \lambda x^{\numbertype{2}}.\ve s:\numbertype{1}\Rightarrow A\wedge\numbertype{2}\Rightarrow A$
and the projection $\pi_1\langle r,s \rangle$ as
$\pi_{\numbertype{1}\Rightarrow A}(\lambda x^{\numbertype{1}}.\ve r \times
\lambda x^{\numbertype{2}}.\ve s)y^{\numbertype{1}}$ (similarly for $\pi_2$), 
where types \numbertype{1} and \numbertype{2} are any two different types. This
example uses free variables, but it is easy to close it, e.g.~use $\lambda y.y$
instead of $y^\numbertype{1}$ in the second line.
Moreover, this technique is not limited to pairs. Due to the associativity
of $\times$, the encoding can be easily extended to lists.

\label{sec:bool}
Example~\ref{ex:bool} on booleans overlooks an interesting fact: If $A\equiv B$,
then both $\true$ and $\false$ behave as a non-deterministic projector. Indeed,
$\true\ve r\ve s\re^*\ve r$, but also $(\lambda x^A.\lambda y^A.x)\ve r\ve s
\eq(\lambda x^A.\lambda y^A.x)(\ve r\times \ve s) \eq(\lambda x^A.\lambda
y^A.x)(\ve s\times \ve r) \eq(\lambda x^A.\lambda y^A.x)\ve s\ve r \re^*\ve s$.
Similarly, $\false\ve r\ve s\re^*\ve s$ and also $\false\ve r\ve s\toreq^*\ve
r$. Hence, $A\Rightarrow A\Rightarrow A$ is not suitable to encode the type
Bool. The type $A\Rightarrow A\Rightarrow A$ has only one term in the underlying
equational theory.

Fortunately, there are ways to construct types with more than one term. First,
let us define the following notation. For any $\ve t$, we write $[\ve
t]^{\tau\Rightarrow\tau}$, the {\em canon} of $\ve t$, that is, the term
$\lambda z^{\tau\Rightarrow\tau}.\ve t$, where $z^{\tau\Rightarrow\tau}$ is a
fresh variable not appearing in $\ve t$. Also, for any term $\ve t$ of type
$(\tau\Rightarrow\tau)\Rightarrow A$, we write $\{\ve
t\}^{\tau\Rightarrow\tau}$, the {\em cocanon}, which is the inverse operation,
that is, $\{[\ve t]^{\tau\Rightarrow\tau}\}^{\tau\Rightarrow\tau}\re\ve t$ for
any $\ve t$ of type $A$. For the cocanon it suffices to take $\{\ve
t\}^{\tau\Rightarrow\tau}=\ve t(\lambda x^\tau.x)$. Therefore, the type
$((\tau\Rightarrow\tau)\Rightarrow A)\Rightarrow A\Rightarrow A$ has the
following two different terms:
$\trtr:=\lambda x^A.\lambda y^{(\tau\Rightarrow\tau)\Rightarrow A}.x$
and
$\ff:=\lambda x^{(\tau\Rightarrow\tau)\Rightarrow A}.\lambda y^A.\{x\}^{\tau\Rightarrow\tau}$.
Hence, it is possible to encode an if-then-else conditional expression as
$\sf{If~}\ve c\sf{~then~}\ve r\sf{~else~}\ve s:= \ve c\ve r[\ve s]^{\tau\Rightarrow\tau}$.
Thus, $ \trtr\ve r[\ve s]^{\tau\Rightarrow\tau} \re^*\ve r
$, while $ \ff\ve r[\ve s]^{\tau\Rightarrow\tau} \eq^*\ff[\ve s]^{\tau\Rightarrow\tau}\ve
r \re^*\{[\ve s]^{\tau\Rightarrow\tau}\}^{\tau\Rightarrow\tau} \re\ve s $.

\section{Conclusion, Discussion and Future Work}\label{sec:conclusion}
In this paper we have defined System I, a proof system for propositional logic, where
isomorphic propositions have the same proofs.

\subsection{Non-terminating extension}\label{sec:deltadelta}
As mentioned in the introduction, the choice of rules is subtle. Indeed, as
well known, the strong normalisation of simply typed lambda calculus is
not a very robust property: minor modifications of typing or reduction rules
can lead to non-terminating calculi, see for instance~\cite{DowekJiangIC11}.
In \OC, we have the rule \rulelabel{(dist$_{\mathsf{app}}$)} to deal with the
equivalence $A\Rightarrow(B\wedge C)\equiv (A\Rightarrow B)\wedge(A\Rightarrow
C)$, and we could have also considered a rule such as $\pi_{A}(rs)\eq\pi_{B\Rightarrow A}(r)s$~\cite{DiazcaroDowekLSFA12}.
However, adding such a rule leads to a non-terminating calculus, as shown
by the following example.
  Let
    $\delta =\lambda
    x^{(\tau\Rightarrow\tau)\wedge\tau}.\pi_{\tau\Rightarrow\tau}(x)\pi_\tau(x):((\tau\Rightarrow\tau)\wedge\tau)\Rightarrow\tau$,\ 
    $\delta'=\delta ((z^{\tau\Rightarrow\tau\Rightarrow\tau}y^\tau)\times
    y^\tau):\tau$, and $\Omega=\delta ((z^{\tau\Rightarrow\tau\Rightarrow\tau}y^\tau)\times\delta'):\tau$.
Then, we have
\begin{align*}
  & \Omega\\
 &\re\pi_{\tau\Rightarrow\tau}((zy)\times\delta')\pi_\tau((zy)\times\delta')
 \re\pi_{\tau\Rightarrow\tau}((zy)\times\delta')\delta'
 =\pi_{\tau\Rightarrow\tau}((zy)\times(\delta((zy)\times y)))\delta'\\
 &\eq\pi_{\tau\Rightarrow\tau}((zy)\times(\delta(zy) y))\delta'
 \eq\pi_{\tau\Rightarrow\tau}((z\times(\delta(zy))) y)\delta'
 \!\!\!\stackrel{\rulelabel{(wrong-rule)}}{\eq^*}\!\!\!\pi_{\tau\Rightarrow\tau}((z\times(\delta(zy))) \delta')y\\
 &\eq\pi_{\tau\Rightarrow\tau}((z\delta')\times(\delta(zy)\delta'))y
 \eq\pi_{\tau\Rightarrow\tau}((z\delta')\times(\delta((zy)\times\delta')))y
 =\pi_{\tau\Rightarrow\tau}((z\delta')\times\Omega)y
\end{align*}

\subsection{Other Related Work}

Apart from the related work already discussed in the introduction,
in a work by Garrigue and A\"it-Kaci~\cite{GarrigueAitkaciPOPL94}, the
selective $\lambda$-calculus has been presented, where only the
isomorphism
\begin{equation}\label{iso:ordering}
  (A\Rightarrow B\Rightarrow C) \equiv (B\Rightarrow A\Rightarrow C).
\end{equation}
has been treated, which is complete with respect to the function type. In \OC we
also consider the conjunction, and hence four isomorphisms. Isomorphism
\eqref{iso:ordering} is a consequence of currification and commutation, that is
$A\wedge B\equiv B\wedge A$ and $(A\wedge B)\Rightarrow C\equiv (A\Rightarrow
B\Rightarrow C)$.

The selective $\lambda$-calculus 
includes
labellings to identify which argument is being used at each time. Moreover, by
considering the Church encoding of pairs, isomorphism \eqref{iso:ordering}
implies isomorphism \eqref{iso:comm} (commutativity of $\wedge$). However, their
proposal is different to ours. In particular, we track the term by its type,
which is a kind of labelling, but when two terms have the same type, then we
leave the system to non-deterministically choose any proof. One of our main
novelties is, indeed, the non-deterministic projector. However, we can also get
back determinism, by encoding a labelling, as discussed in Section
\ref{sec:computing}, or by dropping some isomorphisms (namely,
associativity and commutativity of conjunction).

\subsection{Towards more connectives}
A subtle question is how to add a neutral element of the conjunction, which will
imply more isomorphisms, e.g.~$A\wedge\top\equiv A$,
$A\Rightarrow\top\equiv\top$ and $\top\Rightarrow A\equiv A$. Adding the
equation $\top\Rightarrow\top\equiv\top$ would make it possible to derive
$(\lambda x^{\top}.xx)(\lambda x^{\top}.xx):\top$, however this term is not the
classical $\Omega$, it is typed by $\top$, and imposing some restrictions on the
beta reduction, it could be forced not to reduce to itself but to discard its
argument. For example: ``If $A\equiv\top$, then $(\lambda x^A.\ve r)\ve s\re
r[\star/x]$'', where $\star:\top$ is the introduction rule of $\top$.

\subsection{Eta-expansion rule}
In~\cite{DiazcaroMartinezlopezIFL15} we have given an implementation embedded in
Haskell of an extended fragment of the system as presented
in~\cite{DiazcaroDowekLSFA12}, which is an early version of \OC. In such an
implementation, we have added some rules in order to have only introductions as
normal forms. For example,
``If $s:B$ then $(\lambda x^A.\lambda y^B.r)s\re\lambda x^A.((\lambda y^B.r)s)$.
Such a rule, among others introduced in this implementation, is a particular case of
a more general $\eta$-expansion rule. Indeed, with the rule''
``If $t:A\Rightarrow B$ then $t\re \lambda x^A.tx$''
we can derive
  $(\lambda x^A.\lambda y^B.r)s
  \re
  \lambda z^A.
  (\lambda x^A.\lambda y^B.r)sz
  \eq^*
  \lambda z^A.
  (\lambda x^A.\lambda y^B.r)zs
  \re
  \lambda z^A.((\lambda y^B.r[z/x])s)$.
  
Indeed, we conjecture that \OC extended with an $\eta$-expansion rule would lead
to a system where there is no closed elimination term in normal form.
Such an extension is left for future work.

\bibliography{biblio}

\newpage
\appendix
\section{Detailed proofs of Section~\ref{sec:SR}}\label{app:SR}

\xrecap{Lemma}{Substitution}{lem:substitution}{If $\ve r:A$, $\ve s:B$, and $x\in\V_B$, then $\ve r[\ve s/x]:A$.}
\begin{proof}
  By structural induction on $\ve r$.
  \begin{itemize}
  \item Let $\ve r=x$. By Lemma~\ref{lem:generation}, $A\equiv B$,
    thus $\ve s:A$. We have $x[\ve s/x]=\ve s$, so $x[\ve s/x]:A$.

  \item Let $\ve r=y$, with $y\neq x$. We have $y[\ve s/x]=y$, so
    $y[\ve s/x]:A$.

  \item Let $\ve r=\lambda y^C.\ve r'$. By Lemma~\ref{lem:generation}, $A\equiv
    C\Rightarrow D$, with $\ve r':D$. By the induction
    hypothesis $\ve r'[\ve s/x]:D$, and so, by rule
    $(\Rightarrow_i)$, $\lambda y^C.\ve r'[\ve s/x]:C\Rightarrow
    D$. Since $\lambda y^C.\ve r'[\ve s/x]=(\lambda y^C.\ve r')[\ve s/x]$,
    using rule $(\equiv)$, $(\lambda y^C.\ve r')[\ve s/x]:A$.

  \item Let $\ve r=\ve r_1\ve r_2$. By Lemma~\ref{lem:generation},
    $\ve r_1:C\Rightarrow A$ and $\ve r_2:C$. By
    the induction hypothesis $\ve r_1[\ve s/x]:C\Rightarrow A$ and
    $\ve r_2[\ve s/x]:C$, and so, by rule $(\Rightarrow_e)$,
    $(\ve r_1[\ve s/x])(\ve r_2[\ve s/x]):A$. Since $(\ve
    r_1[\ve s/x])(\ve r_2[\ve s/x])=(\ve r_1\ve r_2)[\ve s/x]$, we have
    $(\ve r_1\ve r_2)[\ve s/x]:A$.

  \item Let $\ve r=\ve r_1\times \ve r_2$. By Lemma~\ref{lem:generation},
    $\ve r_1:A_1$ and $\ve r_2:A_2$, with
    $A\equiv A_1\wedge A_2$. By the induction hypothesis $\ve
    r_1[\ve s/x]:A_1$ and $\ve r_2[\ve s/x]:A_2$, and so, by
    rule $(\wedge_i)$, $(\ve r_1[\ve s/x])\times (\ve r_2[\ve
    s/x]):A_1\wedge A_2$. Since $(\ve r_1[\ve s/x])\times (\ve r_2[\ve
    s/x])=(\ve r_1\times \ve r_2)[\ve s/x]$, using rule $(\equiv)$, we have
    $(\ve r_1\times \ve r_2)[\ve s/x]:A$.

  \item Let $\ve r=\pi_A(\ve r')$. By Lemma~\ref{lem:generation},
    $\ve r':A\wedge C$. Hence, by the induction hypothesis,
    $\ve r'[\ve s/x]:A\wedge C$. Hence, by rule $\wedge_e$,
    $\pi_A(\ve r'[\ve s/x]):A$. Since $\pi_A(\ve r'[\ve
    s/x])=\pi_A(\ve r')[\ve s/x]$, we have $\pi_A(\ve r')[\ve
    s/x]:A$. \qedhere
  \end{itemize}
\end{proof}

\xrecap{Theorem}{Subject reduction}{thm:SR}{If $\ve r:A$ and $\ve r\re\ve s$ or $\ve r\eq\ve s$ then $\ve s:A$.}

\begin{proof}
  By induction on the rewrite relation.
  \begin{itemize}
  \item \rulelabel{(comm)}: 
    If $\ve r\times
    \ve s:A$, then by Lemma~\ref{lem:generation}, $A\equiv A_1\wedge A_2\equiv
    A_2\wedge A_1$, with $\ve r:A_1$ and $\ve s:A_2$.
    Then, $s\times r:A_2\wedge A_1\equiv A$.
  \item \rulelabel{(asso)}: 
    ~
    \begin{description}
    \item[$(^{\to})$] If $(\ve r\times \ve s)\times \ve t:A$, then
      by Lemma~\ref{lem:generation}, $A\equiv (A_1\wedge A_2)\wedge A_3\equiv
      A_1\wedge(A_2\wedge A_3)$, with $\ve r:A_1$, $\ve
      s:A_2$ and $\ve t:A_3$. Then,
      $\ve r\times (\ve s\times \ve t):A_1\wedge(A_2\wedge A_3)\equiv A$.
    \item[$(_{\leftarrow})$] Analogous to $(^{\to})$.
    \end{description}
  \item \rulelabel{(dist$_\lambda$)}:~
    \begin{description}
    \item[$(^{\to})$] If $\lambda x^B.(\ve r\times \ve s):A$, then
      by Lemma~\ref{lem:generation}, $A\equiv (B\Rightarrow(C_1\wedge
      C_2))\equiv ((B\Rightarrow C_1)\wedge(B\Rightarrow C_2))$, with
      $\ve r:C_1$ and $\ve s:C_2$. Then,
${\lambda x^B.\ve r\times \lambda x^B.\ve s:(B\Rightarrow C_1)\wedge(B\Rightarrow C_2)}\equiv A$.
    \item[$(_{\leftarrow})$] If $\lambda x^B.\ve r\times \lambda
      x^B.\ve s:A$, then by Lemma~\ref{lem:generation}, $A\equiv((B\Rightarrow
      C_1)\wedge(B\Rightarrow C_2))\equiv (B\Rightarrow(C_1\wedge C_2))$, with
      $\ve r:C_1$ and $\ve s:C_2$. Then,
${\lambda x^B.(\ve r\times \ve
            s):B\Rightarrow(C_1\wedge C_2)}\equiv A$.
    \end{description}

  \item \rulelabel{(dist$_{\mathsf{app}}$)}:~ 
    \begin{description}
    \item[$(^{\to})$] If $(\ve r\times \ve s)\ve t:A$, then by
      Lemma~\ref{lem:generation}, $\ve r\times \ve s:B\Rightarrow
      A$, and $\ve t:B$. Hence, by Lemma~\ref{lem:generation} again,
      $B\Rightarrow A\equiv C_1\wedge C_2$, and so by Lemma~\ref{lem:ImpConj},
      $A\equiv A_1\wedge A_2$, with $\ve r:B\Rightarrow A_1$ and
      $\ve s:B\Rightarrow A_2$. Then,
   $\ve r\ve t\times \ve s\ve t:A_1\wedge A_2\equiv A$.
    \item[$(_{\leftarrow})$] If $\ve r\ve t\times \ve s\ve t:A$,
      then by Lemma~\ref{lem:generation}, $A\equiv A_1\wedge A_2$ with
      $\ve r:B\Rightarrow A_1$, $\ve s:B'\Rightarrow
      A_2$, $\ve t:B$ and $\ve t:B'$. By
      Lemma~\ref{lem:unicity}, $B\equiv B'$. Then
          ${(\ve r\times \ve s)\ve t:A_1\wedge A_2}\equiv A$.
    \end{description}
  \item \rulelabel{(curry)}:~ 
    \begin{description}
    \item[$(^{\to})$] If $\ve r\ve s\ve t:A$, then by
      Lemma~\ref{lem:generation}, $\ve r:B\Rightarrow C\Rightarrow
      A\equiv(B\wedge C)\Rightarrow A$, $\ve s:B$ and
      $\ve t:C$. Then,
      $r(s\times t):A$.
    \item[$(_{\leftarrow})$] If $\ve r(\ve s\times \ve t):A$, then
      by Lemma~\ref{lem:generation}, $\ve r:(B\wedge C)\Rightarrow
      A\equiv (B\Rightarrow C\Rightarrow A)$, $\ve s:B$ and
      $\ve t:C$. Then $rst:A$.
    \end{description}

  \item \rulelabel{($\beta$)}: 
    If $(\lambda x^B.\ve r)\ve s:A$, then by
    Lemma~\ref{lem:generation}, $\lambda x^B.\ve r:B\Rightarrow A$,
    and by Lemma~\ref{lem:generation} again, $\ve r:A$. Then by
    Lemma~\ref{lem:substitution}, $\ve r[\ve s/x^B]:A$.

  \item \rulelabel{($\pi$)}:
    If
    $\pi_B(\ve r\times \ve s):A$, then by
    Lemma~\ref{lem:generation}, $A\equiv B$, and so, by rule $(\equiv)$,
    $\ve r:A$.

  \item Contextual closure: Let $\ve t\to\ve r$, where $\to$ is either $\eq$ or
    $\re$.
    \begin{itemize}
    \item Let $\lambda x^B.\ve t\to\lambda x^B.\ve r$: If $\lambda
      x^B.\ve t:A$, then by Lemma~\ref{lem:generation}, $A\equiv (B\Rightarrow
      C)$ and $\ve t:C$, hence by the induction hypothesis,
      $\ve r:C$ and so $\lambda x^B.\ve
      r:B\Rightarrow C\equiv A$.
    \item Let $\ve t\ve s\to\ve r\ve s$: If $\ve t\ve s:A$ then by
      Lemma~\ref{lem:generation}, $\ve t:B\Rightarrow A$ and
      $\ve s:B$, hence by the induction hypothesis, $\ve
      r:B\Rightarrow A$ and so $\ve r\ve s:A$.
    \item Let $\ve s\ve t\to\ve s\ve t$:  If $\ve s\ve t:A$ then by
      Lemma~\ref{lem:generation}, $\ve s:B\Rightarrow A$ and
      $\ve t:B$, hence by the induction hypothesis $\ve
      r:B$ and so $\ve s\ve r:A$.
    \item Let $\ve t\times \ve s\to\ve r\times \ve s$:  If $\ve t\times
      \ve s:A$ then by Lemma~\ref{lem:generation}, $A\equiv A_1\wedge A_2$,
      $\ve t:A_1$, and $\ve s:A_2$, hence by the
      induction hypothesis, $\ve r:A_1$ and so $\ve
      r\times \ve s:A_1\wedge A_2\equiv A$.
    \item Let $\ve s\times \ve t\to\ve s\times \ve r$:  Analogous to previous case.
    \item Let $\pi_B(\ve t)\to\pi_B(\ve r)$:  If $\pi_B(\ve t):A$ then
      by Lemma~\ref{lem:generation}, $A\equiv B$ and $\ve t:B\wedge
      C$, hence by the induction hypothesis $\ve r:B\wedge C$.
      Therefore, $\pi_B(\ve r):B\equiv A$.
      \qedhere
    \end{itemize}
  \end{itemize}
\end{proof}
\section{Detailed proofs of Section~\ref{sec:SN}}

\subsection{Detailed proofs of Section~\ref{sec:RP}}\label{app:RP}
\recap{Lemma}{lem:eqProd}{If $\ve r\times \ve s\eq^*\ve t$ then either
  \begin{enumerate}
  \item $\ve t=\ve u\times \ve v$ where either
    \begin{enumerate}
    \item $\ve u\eq^*\ve t_{11}\times \ve t_{21}$ and $\ve v\eq^*\ve
      t_{12}\times \ve t_{22}$ with $\ve r\eq^*\ve t_{11}\times \ve t_{12}$ and
      $\ve s\eq^*\ve t_{21}\times \ve t_{22}$, or
    \item $\ve v\eq^*\ve w\times \ve s$ with $\ve r\eq^*\ve u\times \ve w$, or
      any of the three symmetric cases, or
    \item $\ve r\eq^*\ve u$ and $\ve s\eq^*\ve v$, or the symmetric case.
    \end{enumerate}
  \item $\ve t=\lambda x^A.\ve a$ and $\ve a\eq^*\ve a_1\times \ve
    a_2$ with $\ve r\eq^*\lambda x^A.\ve a_1$ and $\ve s\eq^*\lambda x^A.\ve
    a_2$.
  \item $\ve t=\ve a\ve v$ and $\ve a\eq^*\ve a_1\times \ve a_2$,
    with $\ve r\eq^*\ve a_1\ve v$ and $\ve s\eq^*\ve a_2\ve v$.
  \end{enumerate}
}
\begin{proof}
  By a double induction, first on $M(\ve t)$ and then on the length of the relation
  $\eq^*$. Consider an equivalence proof $\ve r \times \ve s \eq^*\ve t' \eq\ve
  t$ with a shorter proof $\ve r \times \ve s \eq^*\ve t'$. By the second
  induction hypothesis, the term $\ve t'$ has the form prescribed by the lemma.
  We consider the three cases and in each case, the possible rules transforming
  $\ve t'$ in $\ve t$.
  \begin{enumerate}
  \item Let $\ve r\times \ve s\eq^*\ve u\times \ve v\eq\ve t$. The possible
    equivalences from $\ve u\times \ve v$ are
    \begin{itemize}
    \item $\ve t=\ve u'\times \ve v$ or $\ve u\times \ve v'$ with $\ve u\eq\ve
      u'$ and $\ve v\eq\ve v'$, and so the term $\ve t$ is in case \ref{it:eqProd-sum}.
    \item Rules \rulelabel{(comm)} and \rulelabel{(asso)} preserve the conditions of case
      \ref{it:eqProd-sum}.
    \item $\ve t=\lambda x^A.(\ve u'\times \ve v')$, with $\ve u=\lambda x^A.\ve
      u'$ and $\ve v=\lambda x^A.\ve v'$. By the first induction hypothesis
      (since $M(\ve u)<M(\ve t)$ and $M(\ve v)<M(\ve t)$), either
      \begin{enumerate}
      \item $\ve u\eq^*\ve w_{11}\times \ve w_{21}$ and $\ve v\eq^*\ve
        w_{12}\times \ve w_{22}$, by the first induction hypothesis, $\ve
        w_{ij}\eq^*\lambda x^A.\ve t_{ij}$ for $i=1,2$ and $j=1,2$, with $\ve
        u'\eq^*\ve t_{11}\times \ve t_{21}$ and $\ve v'\eq^*\ve t_{12}\times \ve
        t_{22}$, so $\ve u'\times \ve v'\eq^*\ve t_{11}\times \ve t_{12}\times
        \ve t_{21}\times \ve t_{22}$. Hence, $\ve r\eq^*\lambda x^A.(\ve
        t_{11}\times \ve t_{12})$ and $\ve s\eq^*\lambda x^A.(\ve t_{21}\times
        \ve t_{22})$, and hence the term $t$ is in case \ref{it:eqProd-lam}.
      \item $\ve v\eq^*\ve w\times \ve s$ and $\ve r\eq^*\ve u\times \ve w$.
        Since $\ve v\eq^*\lambda x^A.\ve v'$, by the first induction hypothesis,
        $\ve w\eq^*\lambda x^A.\ve t_1$ and $\ve s\eq^*\lambda x^A.\ve t_2$,
        with $\ve v'\eq^*\ve t_1\times \ve t_2$. Hence, $\ve r\eq^*\lambda
        x.(\ve u'\times \ve t_1)$, and hence the term $t$ is in case \ref{it:eqProd-lam}.
      \item $\ve r\eq^*\lambda x^A.\ve u'$ and $\ve s\eq^*\lambda x^A.\ve v$,
        and hence the term $t$ is in case \ref{it:eqProd-lam}.
      \end{enumerate} (the symmetric cases are analogous).
    \item $\ve t=(\ve u'\times \ve v')\ve t'$, with $\ve u=\ve u'\ve t'$ and
      $\ve v=\ve v'\ve t'$. By the first induction hypothesis (since $M(\ve
      u)<M(\ve t)$ and $M(\ve v)<M(\ve t)$), either
      \begin{enumerate}
      \item $\ve u\eq^*\ve w_{11}\times \ve w_{21}$, $\ve v\eq^*\ve w_{12}\times
        \ve w_{22}$, $\ve r\eq^*\ve w_{11} \times \ve w_{12}$, and $\ve s \eq^*
        \ve w_{21} \times \ve w_{22}$. By the first induction hypothesis (since
        $M(\ve w_{ij})<M(\ve t)$), $\ve w_{ij}\eq^*\ve t_{ij}\ve t'$, for
        $i=1,2$ and $j=1,2$, where $\ve u'\eq^*\ve t_{11}\times \ve t_{21}$,
        $\ve v'\eq^*\ve t_{12}\times \ve t_{22}$, $\ve r\eq^*\ve w_{11}\times
        \ve w_{12}$ and $\ve s\eq^*\ve w_{21}\times \ve w_{22}$. Therefore, $\ve
        u\eq^*\ve t_{11}\ve t'\times \ve t_{21}\ve t'$ and $\ve v\eq^*\ve
        t_{12}\ve t'\times \ve t_{22}\ve t'$ with $\ve r\eq^*(\ve t_{11}\times
        \ve t_{12})\ve t'$ and $\ve s\eq^*(\ve t_{21}\times \ve t_{22})\ve t'$,
        and hence the term $t$ is in case \ref{it:eqProd-app}.
      \item $\ve v'\ve t'\eq^*\ve w\times \ve s$ and $\ve r\eq^*\ve u'\ve
        t'\times \ve w$. By the first induction hypothesis on $\ve v'\ve
        t'\eq^*\ve w\times \ve s$ (since $M(\ve w)<M(\ve t)$ and $M(\ve s)<M(\ve
        t)$), we have $\ve w\eq^*\ve t_1\ve t'$ and $\ve s\eq^*\ve t_2\ve t'$
        with $\ve t_1\times \ve t_2\eq^*\ve v'$. Therefore, $\ve v\eq^*\ve
        t_1\ve t'\times \ve t_2\ve t'$ with $\ve r\eq^*(\ve u\times \ve t_1)\ve
        t'$ and $\ve s\eq^*\ve t_2\ve t'$, and $\ve u'\times \ve v'\eq^*\ve
        u'\times \ve t_1\times \ve t_2$, hence the term $t$ is in case \ref{it:eqProd-app}.
      \item $\ve r\eq^*\ve u'\ve t'$ and $\ve s\eq^*\ve v'\ve t'$, and hence we
        are in case \ref{it:eqProd-app}.
      \end{enumerate} (the symmetric cases are analogous).
    \end{itemize}
  \item Let $\ve r\times \ve s\eq^*\lambda x^A.\ve a\eq\ve t$, with $\ve
    a\eq^*\ve a_1\times \ve a_2$, $\ve r\eq^*\lambda x^A.\ve a_1$, and $\ve
    s\eq^*\lambda x^A.\ve a_2$. Hence, possible equivalences from $\lambda
    x.\ve a$ to $\ve t$ are
    \begin{itemize}
    \item $\ve t=\lambda x^A.\ve a'$ with $\ve a\eq^*\ve a'$, hence $\ve
      a'\eq^*\ve a_1\times \ve a_2$, and so the term $t$ is in case \ref{it:eqProd-lam}.
    \item $\ve t=\lambda x^A.\ve u\times \lambda x^A.\ve v$, with $\ve a_1\times
      \ve a_2\eq^*\ve a=\ve u\times \ve v$. Hence, by the first induction
      hypothesis (since $M(\ve a)<M(\ve t)$), either
      \begin{enumerate}
      \item $\ve a_1\eq^*\ve u$ and $\ve a_2\eq^*\ve v$, and so $\ve
        r\eq^*\lambda x^A.\ve u$ and $\ve s\eq^*\lambda x^A.\ve v$, or
      \item $\ve v\eq^*\ve t_1\times \ve t_2$ with $\ve a_1\eq^*\ve u\times \ve
        t_1$ and $\ve a_2\eq^*\ve t_2$, and so $\lambda x^A.\ve v\eq^*\lambda
        x.\ve t_1\times \lambda x^A.\ve t_2$, $\ve r\eq^*\lambda x^A.\ve
        u\times \lambda x^A.\ve t_1$ and $\ve s\eq^*\lambda x^A.\ve t_2$, or
      \item $\ve u\eq^*\ve t_{11}\times \ve t_{21}$ and $\ve v\eq^*\ve
        t_{12}\times \ve t_{22}$ with $\ve a_1\eq^*\ve t_{11}\times \ve t_{12}$
        and $\ve a_2\eq^*\ve t_{21}\times \ve t_{22}$, and so $\lambda x^A.\ve
        u\eq^*\lambda x^A.\ve t_{11}\times \lambda x^A.\ve t_{21}$, $\lambda
        x.\ve v\eq^*\lambda x^A.\ve t_{12}\times \lambda x^A.\ve t_{22}$, $\ve
        r\eq^*\lambda x^A.\ve t_{11}\times \lambda x^A.\ve t_{12}$ and $\ve
        s\eq^*\lambda x^A.\ve t_{21}\times \lambda x^A.\ve t_{22}$.
      \end{enumerate} (the symmetric cases are analogous), and so the term $t$
      is in case \ref{it:eqProd-sum}.
    \end{itemize}
  \item Let $\ve r\times \ve s\eq^*\ve a\ve w\eq\ve t$, with $\ve a\eq^*\ve
    a_1\times \ve a_2$, $\ve r\eq^*\ve a_1\ve w$, and $\ve s\eq^*\ve a_2\ve w$.
    The possible equivalences from $\ve a\ve w$ to $\ve t$ are
    \begin{itemize}
    \item $\ve t=\ve a'\ve w$ with $\ve a\eq^*\ve a'$, hence $\ve a'\eq^*\ve
      a_1\times \ve a_2$, and so the term $t$ is in case \ref{it:eqProd-app}.
    \item $\ve t=\ve a\ve w'$ with $\ve w\eq^*\ve w'$ and so the term $t$ is in case \ref{it:eqProd-app}.
    \item $\ve t=\ve u\ve w\times \ve v\ve w$, with $\ve a_1\times \ve
      a_2\eq^*\ve a=\ve u\times \ve v$. Hence, by the first induction hypothesis
      (since $M(\ve a)<M(\ve t)$), either
      \begin{enumerate}
      \item $\ve a_1\eq^*\ve u$ and $\ve a_2\eq^*\ve v$, and so $\ve r\eq^*\ve
        u\ve w$ and $\ve s\eq^*\ve v\ve w$, or
      \item $\ve v\eq^*\ve t_1\times \ve t_2$ with $\ve a_1\eq^*\ve u\times \ve
        t_1$ and $\ve a_2\eq^*\ve t_2$, and so $\ve v\ve w\eq^*\ve t_1\ve
        w\times \ve t_2\ve w$, $\ve r\eq^*\ve u\ve w\times \ve t_1\ve w$ and
        $\ve s\eq^*\ve t_2\ve w$, or
      \item $\ve u\eq^*\ve t_{11}\times \ve t_{21}$ and $\ve v\eq^*\ve
        t_{12}\times \ve t_{22}$ with $\ve a_1\eq^*\ve t_{11}\times \ve t_{12}$
        and $\ve a_2\eq^*\ve t_{21}\times \ve t_{22}$, and so $\ve u\ve
        w\eq^*\ve t_{11}\ve w\times \ve t_{21}\ve w$, $\ve v\ve w\eq^*\ve
        t_{12}\ve w\times \ve t_{22}\ve w$, $\ve r\eq^*\ve t_{11}\ve w\times \ve
        t_{12}\ve w$ and $\ve s\eq^*\ve t_{21}\ve w\times \ve t_{22}\ve w$.
      \end{enumerate} (the symmetric cases are analogous), and so the term $t$
      is in case \ref{it:eqProd-sum}.
    \item $\ve t=\ve a'(\ve v\times \ve w)$ with $\ve a=\ve a'\ve v$, thus $\ve
      a'\ve v=\ve a\eq^*\ve a_1\times \ve a_2$. Hence, by the first induction
      hypothesis, $\ve a'\eq^*\ve a'_1\times \ve a'_2$, with $\ve a_1\eq^*\ve
      a'_1\ve v$ and $\ve a_2\eq^*\ve a'_2\ve v$. Therefore, $\ve r\eq^*\ve
      a'_1(\ve v\times \ve w)$ and $\ve s\eq^*\ve a'_2(\ve v\times \ve w)$, and
      so the term $t$ is in case \ref{it:eqProd-app}.
      \qedhere
    \end{itemize}
  \end{enumerate}
\end{proof}

\recap{Lemma}{lem:reProd}{If $\ve r_1\times \ve r_2\eq^*\ve s\re\ve t$, there exists $\ve u_1$, $\ve u_2$ such that $\ve t\eq^*\ve u_1\times \ve u_2$ and either ($\ve r_1\toreq\ve u_1$ and $\ve r_2\toreq\ve u_2$), or ($\ve r_1\toreq\ve u_1$ and $\ve r_2\eq^*\ve u_2$), or ($\ve r_1\eq^*\ve u_1$ and $\ve r_2\toreq\ve u_2$).}
\begin{proof}
  By induction on $M(\ve r_1\times \ve r_2)$. By Lemma~\ref{lem:eqProd}, $\ve s$
  is either a product, an abstraction or an application with the conditions
  given in the lemma. The different terms $\ve s$ reducible by $\re$ are
  \begin{itemize}
  \item $(\lambda x^A.\ve a)\ve s'$ that reduces by the \rulelabel{($\beta$)} rule to $\ve
    a[\ve s'/x]$.
  \item $\ve s_1\times \ve s_2$, $\lambda x^A.\ve a$, $\ve a\ve s'$, with a
    reduction in the subterm $\ve s_1$, $\ve s_2$, $\ve a$, or $\ve s'$.
  \end{itemize}
  Notice that rule \rulelabel{($\pi$)} cannot apply since $\ve
  s\not\eq^*\pi_C(\ve s')$.

  We consider each case:
  \begin{itemize}
  \item $\ve s=(\lambda x^A.\ve a)\ve s'$ and $\ve t=\ve a[\ve s'/x]$. Using
    twice Lemma~\ref{lem:eqProd}, we have $\ve a\eq^*\ve a_1\times \ve a_2$,
    $\ve r_1\eq^*(\lambda x^A.\ve a_1)\ve s'$ and $\ve r_2\eq^*(\lambda x^A.\ve
    a_2)\ve s'$. Since $\ve t\eq^*\ve a_1[\ve s'/x]\times \ve a_2[\ve
    s'/x]$, we take $\ve u_1=\ve a_1[\ve s'/x]$ and $\ve u_2=\ve a_2[\ve
    s'/x]$.
  \item $\ve s=\ve s_1\times \ve s_2$, $\ve t=\ve t_1\times \ve s_2$ or $\ve
    t=\ve s_1\times \ve t_2$, with $\ve s_1\re\ve t_1$ and $\ve s_2\re\ve t_2$.
    We only consider the first case since the other is analogous. One of the
    following cases happen
    \begin{enumerate}
    \item[(a)] $\ve r_1\eq^*\ve w_{11}\times \ve w_{21}$, $\ve r_2\eq^*\ve
      w_{12}\times \ve w_{22}$, $\ve s_1=\ve w_{11}\times \ve w_{12}$ and $\ve
      s_2=\ve w_{21}\times \ve w_{22}$. Hence, by the induction hypothesis,
      either $\ve t_1=\ve w'_{11}\times \ve w_{12}$, or $\ve t_1=\ve
      w_{11}\times \ve w'_{12}$, or $\ve t_1=\ve w'_{11}\times \ve w'_{12}$,
      with $\ve w_{11}\re\ve w_{11}'$ and $\ve w_{12}\re\ve w_{12}'$. We take,
      in the first case $\ve u_1=\ve w_{11}'\times \ve w_{21}$ and $\ve u_2=\ve
      w_{12}\times \ve w_{22}$, in the second case $\ve u_1=\ve w_{11}\times \ve
      w_{21}$ and $\ve u_2=\ve w_{12}'\times \ve w_{22}$, and in the third $\ve
      u_1=\ve w'_{11}\times \ve w_{21}$ and $\ve u_2=\ve w_{12}'\times \ve
      w_{22}$.
    \item[(b)] We consider two cases, since the other two are symmetric.
      \begin{itemize}
      \item $\ve r_1\eq^*\ve s_1\times \ve w$ and $\ve s_2\eq^*\ve w\times \ve
        r_2$, in which case we take $\ve u_1=\ve t_1\times \ve w$ and $\ve
        u_2=\ve r_2$.
      \item $\ve r_2\eq^*\ve w\times \ve s_2$ and $\ve s_1=\ve r_1\times \ve w$.
        Hence, by the induction hypothesis, either $\ve t_1=\ve r'_1\times \ve
        w$, or $\ve t_1=\ve r_1\times \ve w'$ or $\ve t_1=\ve r'_1\times \ve
        w'$, with $\ve r_1\re\ve r_1'$ and $\ve w\re\ve w'$. We take, in the
        first case $\ve u_1=\ve r'_1$ and $\ve u_2=\ve w\times \ve s_2$, in the
        second case $\ve u_1=\ve r_1$ and $\ve u_2=\ve w'\times \ve s_2$, and in
        the third case $\ve u_1=\ve r'_1$ and $\ve u_2=\ve w'\times \ve s_2$.
      \end{itemize}
    \item[(c)] $\ve r_1\eq^*\ve s_1$ and $\ve r_2\eq^*\ve s_2$, in which case we
      take $\ve u_1=\ve t_1$ and $\ve u_2=\ve s_2$.

    \end{enumerate}
  \item $\ve s=\lambda x^A.\ve a$, $\ve t=\lambda x^A.\ve t'$, and $\ve a\re\ve
    t'$, with $\ve a\eq^*\ve a_1\times \ve a_2$ and $\ve s\eq^*\lambda x^A.\ve
    a_1\times \lambda x^A.\ve a_2$. Therefore, by the induction hypothesis,
    there exists $\ve u'_1$, $\ve u'_2$ such that either ($\ve a_1\toreq\ve
    u'_1$ and $\ve a_2\toreq\ve u'_2$), or ($\ve a_1\eq^*\ve u'_1$ and $\ve
    a_2\toreq\ve u'_2$), or ($\ve a_1\toreq\ve u'_1$ and $\ve a_2\eq^*\ve
    u'_2$). Therefore, we take $\ve u_1=\lambda x^A.\ve u_1'$ and $\ve
    u_2=\lambda x^A.\ve u_2'$.
  \item $\ve s=\ve a\ve s'$, $\ve t=\ve t'\ve s'$, and $\ve a\re\ve t'$, with
    $\ve a\eq^*\ve a_1\times \ve a_2$ and $\ve s\eq^*\ve a_1\ve s'\times \ve
    a_2\ve s'$. Therefore, by the induction hypothesis, there exists $\ve u'_1$,
    $\ve u'_2$ such that either ($\ve a_1\toreq\ve u'_1$ and $\ve a_2\toreq\ve
    u'_2$), or ($\ve a_1\eq^*\ve u'_1$ and $\ve a_2\toreq\ve u'_2$), or ($\ve
    a_1\toreq\ve u'_1$ and $\ve a_2\eq^*\ve u'_2$). Therefore, we take $\ve
    u_1=\ve u_1'\ve s'$ and $\ve u_2=\ve u_2'\ve s'$.
  \item $\ve s=\ve a\ve s'$, $\ve t=\ve a\ve t'$, and $\ve s'\re\ve t'$, with
    $\ve a\eq^*\ve a_1\times \ve a_2$ and $\ve s\eq^*\ve a_1\ve s'\times \ve
    a_2\ve s'$. By Lemma~\ref{lem:eqProd} several times, one the following cases
    happen
    \begin{enumerate}
    \item[(a)] $\ve a_1\ve s'\eq^*\ve w_{11}\ve s'\times \ve w_{12}\ve s'$, $\ve
      a_2\ve s'\eq^*\ve w_{21}\ve s'\times \ve w_{22}\ve s'$, $\ve r_1\eq^*\ve
      w_{11}\ve s'\times \ve w_{21}\ve s'$ and $\ve r_2\eq^*\ve w_{12}\ve
      s'\times \ve w_{22}\ve s'$. We take $\ve u_1\eq^*(\ve w_{11}\times \ve
      w_{21})\ve t'$ and $\ve r_2\eq^*(\ve w_{12}\times \ve w_{22})\ve t'$.
    \item[(b)] $\ve a_2\ve s'\eq^*\ve w_1\ve s'\times \ve w_2\ve s'$, $\ve
      r_1\eq^*\ve a_1\ve s'\times \ve w_2\ve s'$ and $\ve r_2\eq^*\ve w_2\ve
      s'$. So we take $\ve u_1=(\ve a_1\times \ve w_1)\ve t'$ and $\ve u_2=\ve
      w_2\ve t'$, the symmetric cases are analogous.
    \item[(c)] $\ve r_1\eq^*\ve a_1\ve s'$ and $\ve r_2\eq^*\ve a_2\ve s'$, in
      which case we take $\ve u_1=\ve a_1\ve t'$ and $\ve u_2=\ve a_2\ve t'$ the
      symmetric case is analogous. \qedhere
    \end{enumerate}
  \end{itemize}
\end{proof}

\subsection{Detailed proofs of Section~\ref{sec:Adequacy}}\label{app:Adequacy}
\xrecap{Lemma}{Adequacy of product}{lem:adequacyOfProd}{If $r\in\interp A$ and $s\in\interp B$, then $\ve r\times \ve s\in\interp{A\wedge B}$.}
\begin{proof}
  We need to prove that $K^\tau_{A\wedge B}[r\times s]\in\SN$.
  We proceed by induction on the number of projections in $K^\tau_{A\wedge B}$.
  Since the hole of $K^\tau_{A\wedge B}$ has type $A\wedge B$, and $K^\tau_{A\wedge B}[t]$ has type
  $\tau$ for any $t:A$ there is at least one projection.
  
  As $r\in\interp A$, for any elimination context $K'^\tau_A$, we have
  $K'^\tau_A[r]\in\SN$, but then if $A\equiv B_1\Rightarrow\dots\Rightarrow
  B_n\Rightarrow C$, we also have $K''^\tau_C[rt_1\dots t_n]\in\SN$, thus
  $rt_1\dots t_n\in\interp C$.
  Similarly, since $s\in\interp B$, $st_1\dots t_n$ is reducible.
  
  We prove that $K'^\tau_C[\pi_C (rt_1\dots t_n\times st_1\dots t_n)]\in\SN$
  by showing, more generally, that if $r'$ and $s'$ are two reducts of
  $rt_1\dots t_n$ and $st_1\dots t_n$, then
  $K'^\tau_C[\pi_C(r'\times s')]\in\SN$. For this, we show that all its one step reducts are
  in $\SN$, by induction on $|K'^\tau_C| + |r'|+|s'|$.
  \begin{itemize}
  \item If the reduction takes place in one of the terms in $\T(K'^\tau_C)$, in
    $r'$, or in $s'$, we apply the induction hypothesis.
  \item Otherwise, the reduction is a \rulelabel{($\pi$)} reduction of $\pi_C
    (r'\times s')$, that is, $r'\times s'\eq^*v\times w$, the reduct is $v$, and we need to prove
    $K'^\tau_C[v]\in\SN$.
    By Lemma~\ref{lem:eqProd}, we have either:
    \begin{itemize}
    \item $v\eq^*r_1\times s_1$, with $r'\eq^*r_1\times r_2$ and
      $s'\eq^*s_1\times s_2$. In such a case,
      by Lemma~\ref{lem:SNimpliespiSN}, $v$ is the product of two reducible
      terms, so since there is one projection less than in $K^\tau_{A\wedge B}$, the first induction hypothesis applies.
    \item $v\eq^*r'\times s_1$, with $s'\eq^*s_1\times
      s_2$. In such a case,
      by Lemma~\ref{lem:SNimpliespiSN}, $v$ is the product of two reducible terms, so since there is one projection less than in $K^\tau_{A\wedge B}$, the first induction hypothesis applies.
    \item $v\eq^*r_1\times s'$, with $r'\eq^*r_1\times
      r_2$. In such a case,
      by Lemma~\ref{lem:SNimpliespiSN}, $v$ is the product of two reducible terms, so since there is one projection less than in $K^\tau_{A\wedge B}$, the first induction hypothesis applies.
    \item $v\eq^*r'$, in which case, $C\equiv A$, and since $r\in\interp A$, we have $K'^\tau_A[r']\eq^*K'^\tau_A[v]\in\SN$.
    \item $v\eq^*r_1$ with $r'\eq^*r_1\times r_2$, in which case, since $r\in\interp A$, we have $K'^\tau_C[\pi_C(r')]\in\SN$ and $K'^\tau_C[\pi_C(r')]\toreq K'^\tau_C[v]$ hence $K'^\tau_C[v]\in\SN$.
    \item $v\eq^*s'$, in which case, $C\equiv B$, and since $s\in\interp B$, we have $K'^\tau_B[s']\eq^*K'^\tau_B[v]\in\SN$.
    \item $v\eq^*s_1$ with $s'\eq^*s_1\times s_2$, in which case, since $s\in\interp B$, we have $K'^\tau_C[\pi_C(s')]\in\SN$ and $K'^\tau_C[\pi_C(s')]\toreq K'^\tau_C[v]$ hence $K'^\tau_C[v]\in\SN$.
      \qedhere
    \end{itemize}
  \end{itemize}
\end{proof}

\xrecap{Lemma}{Adequacy of abstraction}{lem:lamOfInterp}{If for all $t\in\interp A$, $r[t/x]\in\interp B$, then $\lambda x^A.r\in\interp{A\Rightarrow B}$.}
\begin{proof}
  We proceed by induction on $M(r)$.
  
  If $r\eq^*r_1\times r_2$, by
  Lemma~\ref{lem:generation}, $B\equiv B_1\wedge B_2$ with $r_1:B_1$ and $r_2:B_2$.
  and so by Lemma~\ref{lem:substitution},
  $r_1[t/x]:B_1$ and $r_2[t/x]:B_2$.
  Since $r[t/x]\in\interp B$, we have $r_1[t/x]\times r_2[t/x]\in\interp{B}$.
  By
  Lemma~\ref{lem:SNimpliespiSN}, $r_1[t/x]\in\interp{B_1}$ and
  $r_2[t/x]\in\interp{B_2}$. By the induction hypothesis, $\lambda
  x^A.r_1\in\interp{A\Rightarrow B_1}$ and $\lambda
  x^A.r_2\in\interp{A\Rightarrow B_2}$, then by Lemma~\ref{lem:adequacyOfProd}, 
$\lambda
  x^A.r\eq^*\lambda x^A.r_1\times\lambda x^A.r_2\in\interp{(A\Rightarrow
    B_1)\wedge(A\Rightarrow B_2)}$, and by Lemma~\ref{lem:eqInterp},
  $\interp{(A\Rightarrow B_1)\wedge(A\Rightarrow B_2)}=\interp{A\Rightarrow B}$.

  If $r\not\eq^* r_1\times r_2$, we need to prove that for any elimination context $K^\tau_{A\Rightarrow B}$, we have
  $K^\tau_{A\Rightarrow B}[\lambda x^A. r]\in\SN$.
  
  Since $r$ and all the terms in $\T(K^\tau_{A\Rightarrow B})$ are reducible,
  then they are in $\SN$, by Lemma~\ref{lem:CR1}.
  We proceed by induction on $|K^\tau_{A\Rightarrow B}| + |r|$ to show that all
  the one step reducts of $K^\tau_{A\Rightarrow B}[\lambda x^A.r]$ are in $\SN$.
  Since $r$ is not a product, the only one step reducts are the following.
  \begin{itemize}
  \item If the reduction takes place in one of the terms in $\T(K^\tau_{A\Rightarrow B})$ or $r$, we apply the
    induction hypothesis.
  \item If $K^\tau_{A\Rightarrow B}[\lambda x^A.r] = K'^\tau_{B}[(\lambda x^A.r)
s]$ and it reduces to $K'^\tau_{B}[r[s/x]]$, as $r[s/x]\in\interp B$, we have
$K'^\tau_B[ r[s/x]]\in\SN$.
    \qedhere
  \end{itemize}
\end{proof}

\xrecap{Theorem}{Adequacy}{thm:adequacy}{If $\ve r:A$, then for all $\sigma$ adequate, we have $\sigma r\in\interp A$.}
\begin{proof}
  By induction on $r$. 
  \begin{itemize}
  \item If $r$ is a variable $x\in\V_A$, then, since $\sigma$ is adequate,
    we have $\sigma r\in\interp A$.
    
  \item If $r$ is a product $s\times t$, then by Lemma~\ref{lem:generation}, $s:B$, $t:C$, and $A\equiv
    B\wedge C$, then by the induction
    hypothesis, $\sigma s\in\interp B$ and $\sigma t\in\interp C$. By
    Lemma~\ref{lem:adequacyOfProd}, ${(\sigma s\times\sigma t)}\in\interp{B\wedge C}$,
    hence by Lemma~\ref{lem:eqInterp}, $\sigma r\in\interp A$.
    
  \item If $r$ is a projection $\pi_A(s)$, then by Lemma~\ref{lem:generation},
$s:A\wedge B$, and by the induction hypothesis, $\sigma
s\in\interp{A\wedge B}$. By Lemma~\ref{lem:SNimpliespiSN}, $\pi_A(\sigma
s)\in\interp{A}$, hence $\sigma r\in\interp A$.
    
  \item If $r$ is an abstraction $\lambda x^B.s$, with $s:C$, then by
    Lemma~\ref{lem:generation}, $A\equiv
    B\Rightarrow C$, hence by the induction
    hypothesis, for all $\sigma$, and for all $t\in\interp B$,  $(\sigma s)[t/x]\in\interp C$.
    Hence, by Lemma~\ref{lem:lamOfInterp}, ${\lambda x^B.\sigma
      s}\in\interp{B\Rightarrow C}$, 
    hence, by Lemma~\ref{lem:eqInterp}, $\sigma r\in\interp A$.
    
  \item If $r$ is an application $st$, then by Lemma~\ref{lem:generation},
    $s:B\Rightarrow A$ and $t:B$, then by the induction hypothesis,
    $\sigma s\in\interp{B\Rightarrow A}$ and 
    $\sigma t\in\interp B$. Hence, by
    Lemma~\ref{lem:appOfInterp}, we have $\sigma r=\sigma s\sigma
    t\in\interp A$.
    \qedhere
  \end{itemize}
\end{proof}

\end{document}